\documentclass[10pt]{article}

\usepackage{ragged2e}
\usepackage{times}
\usepackage{amsmath}
\usepackage{amssymb}
\usepackage{geometry}
\usepackage{graphicx}
\usepackage{caption}
\usepackage{subcaption}
\usepackage{float}
\usepackage{epstopdf}
\usepackage{algpseudocode}
\usepackage{algorithm}
\usepackage{enumerate}   
\DeclareCaptionFont{tiny}{\tiny}

\newtheorem{theorem}{Theorem}[section]
\newtheorem{defi}{Definition}[section]
\newtheorem{lemma}[theorem]{Lemma}
\newtheorem{claim}[theorem]{Claim}

\newcommand{\sq}{\hbox{\rlap{$\sqcap$}$\sqcup$}}
\newcommand{\qed}{\hspace*{\fill}\sq}
\newenvironment{proof}{\noindent {\bf Proof.}\ }{\qed\par\vskip 4mm\par}

\begin{document}

\title{A Primal-Dual based  Distributed Approximation Algorithm for Prize-Collecting Steiner Tree}

\author{Parikshit Saikia$^{1,*}$, Sushanta Karmakar$^{1,\dag}$, and Aris T. Pagourtzis$^{2,\ddag}$\\
$^{*}$s.parikshit@iitg.ernet.in, $^{\dag}$sushantak@iitg.ernet.in, and $^{\ddag}$pagour@cs.ntua.gr\\
		$^{1}$Department of Computer Science and Engineering\\
		Indian Institute of Technology Guwahati, India, 781039\\
		\and
		$^{2}$ School of Electrical and Computer Engineering\\
		National Technical University of Athens\\
		Politechniou 9, GR-15780 Zographou, Greece\\
	}

\maketitle \thispagestyle{empty}


\begin{abstract}
	The Prize-Collecting Steiner Tree (PCST) problem is a generalization of  the Steiner Tree problem that has applications in network design, content distribution networks, and many more. There are a few centralized approximation algorithms \cite{DB_MG_DS_DW_1993, GW_1995, DJ_MM_SP_2000, AA_MB_MH_2011} for solving the PCST problem. However no distributed algorithm is known that solves PCST with a guaranteed approximation factor. In this work we present an asynchronous distributed  $(2 - \frac{1}{n - 1})$-approximation algorithm that constructs a PCST for a given connected undirected graph with non-negative edge weights and a non-negative prize value for each node. Our algorithm is an adaptation of the centralized algorithm proposed by Goemans and Williamson \cite{GW_1995} to the distributed setting, and is based on the primal-dual method. The message complexity of the algorithm with input graph having node set $V$ and edge set $E$ is $O(|V||E|)$. Initially each node  knows only its own prize value and the weight of each incident edge. The algorithm is spontaneously initiated at a special node called the \emph{root node} and when it terminates each node knows whether it is in the PCST or not. To the best of our knowledge this is the first distributed constant approximation algorithm for PCST.
\end{abstract}


\centerline{{\bf Keywords}: Steiner Tree, Prize-Collecting Steiner Tree, Distributed Approximation, Primal-Dual.}



\section{Introduction} \label{intro}
\vspace{-.6em} 
The Minimum Spanning Tree (MST) problem is a fundamental problem in graph theory and network design. Given a connected graph $G=(V,E)$ and a weight function $w : E \rightarrow \mathbb{R}^{+}$, the goal of the MST problem is to find a subgraph $G'=(V,E')$ of $G$ connecting all vertices of $V$ such that $\sum_{e \in E'} w_e$ is minimized. There are many centralized \cite{citeulike:4031585, BLTJ:BLTJ1515} and distributed algorithms \cite{GHS_1983, Faloutsos:1995:ODA:224964.225474} for MST construction. 
Steiner Tree (ST) problem is a generalization of the MST problem. The definition of ST is as follows: given a connected graph $G=(V,E)$ and a weight function $w : E \rightarrow \mathbb{R}^{+}$, and a set of vertices $Z \subseteq V$, known as the set of terminals, the goal of the ST problem is to find a subgraph $G'=(V',E')$ of $G$ such that $\sum_{e \in E'} w_e$ is minimized subject to the condition that $Z \subseteq V' \subseteq V$.  

Both MST and ST problems have many applications in VLSI layout design \cite{JK4069504}, communication networks~\cite{Du:2008:STP:1628718}, transportation networks~\cite{TM_RW_1984} etc. It is known that the MST problem can be solved in polynomial time, however the ST problem is NP-hard~\cite{DBLP:conf/coco/Karp72}. Therefore many polynomial time approximation algorithms have been proposed for the ST problem \cite{Zelikovsky1993, BERMAN1994381, PROMEL200089, Karpinski1997, Robins:2000:IST:338219.338638} with various approximation ratios and complexities. Byrka et al. \cite{Byrka:2010:ILA:1806689.1806769} proposed a polynomial time approximation algorithm for the ST problem for a general graph which has the best known  approximation factor of $\ln 4 + \epsilon \approx 1.386 + \epsilon$, for $\epsilon > 0$. It is a centralized algorithm that uses the technique of {\em iterative randomized rounding} of LPs. It is also known that the ST problem for general graphs cannot be solved in polynomial time with an approximation factor $\leq \frac{96}{95}$ \cite{chlebik:2008:STP:1414105.1414423}. There are many variations of the ST problem such as Directed Steiner Tree \cite{Zosin:2002:DST:545381.545388, AA_AF_BG_2016, Charikar:1998:AAD:314613.314700,  Watel:2013:SPL:2694605.2694640},  Metric Steiner Tree \cite{Robins:2005:TBG:1068396.1071708, Karpinski1997}, Euclidean Steiner Tree \cite{Arora:1998:PTA:290179.290180}, Rectilinear Steiner Tree \cite{Hanan_1966, Hwang_1976, Zelikovsky1993, BERMAN1994381, Karpinski1997}, Steiner Forest \cite{AA_PK_RR_1995, CL_BPS_2014, GW_1995, Williamson1995} and so on. Hauptman and Karpinaski \cite{Hauptman-Karpinaski} provide a website with continuously updated state of the art results for many variants of the problem.  Out of the many variants, we focus on a generalization of the ST problem called ``Prize-Collecting Steiner Tree problem''.

\begin{defi}[{Prize-Collecting Steiner Tree (PCST) problem}]
	Given a connected weighted graph $G = (V, E, p, w)$ where $V$ is the set of vertices, $E$ is the set of edges, $p : V \rightarrow \mathbb{R}^+$ is a non-negative prize function and $w : E \rightarrow  \mathbb{R}^+$ is a non-negative weight function, the goal is to find a tree $T = (V', E')$ where $V' \subseteq V$ and $E' \subseteq E$ that minimizes the following function:
	\begin{align*}
		\hspace{12em} GW(T) = \sum_{e \in E'} w_e + \sum_{v \notin V'} p_v
	\end{align*}
\end{defi}

A feasible solution to the PCST problem has two parts namely \emph{Steiner} and \emph{Penalty}. A node is in the Steiner part if it is covered by $T$, otherwise it belongs to the Penalty part. This problem has applications in situations where various demand points (nodes) need to form a structure with minimum total connection cost. Each demand point has some non-negative prize associated with it. If some of the demand points are too expensive to connect then it might be better not to include them in the structure and instead lose their prize---or, equivalently, pay a penalty, equal or proportional to their prize value. The goal is to minimize the overall cost, considered as the sum of connection costs plus lost prizes. Note that the Steiner Tree problem is in fact a special case of PCST, where we set the prize of terminals to $\infty$ and the prize of all other nodes to $0$; therefore the PCST problem is NP-hard.  

\smallskip
\noindent
{\bf Motivation.} The PCST problem has many  practical applications in network design (e.g. rail-road networks, optical fibre networks), content distribution networks (video on demand, streaming multicast) etc. For example, suppose a company wants to build an optical fibre network to provide broadband internet service to various customers. Here the graph might correspond to a street map, since optical fibre is typically laid along the streets. In this case street segments can be considered as the set of edges, and street intersections and street endpoints as the set of vertices (aka potential customer locations).  The cost of the edges are the installation costs of the cables. The prize associated with the vertex representing a customer is an estimate of revenue obtained by connecting the customer to the network. Vertices corresponding to the intersections other than the customers have zero prize.  An optimal design of this network needs to take care of two objectives, \textit{(i) connect to a set of customers that maximizes the profit} and \textit{(ii) connection cost is minimum}. To achieve this, some customers may be excluded from the structure as they may incur more connection cost. For such customers, the company pays a penalty which is proportional to the prize of the node. Therefore, the overall goal is to decide a subset of customers that should be connected so that the sum of the connection cost and the total penalty (for not connecting some customers) is minimized. This situation can be modelled as the PCST problem. Note that PCST equivalently captures the case where prizes are payments by customers and the objective is to maximize the company's net profit (payments minus connection cost). Similarly many other practical problems in {\em protein-protein interaction network} \cite{Dittrich_2008}, {\em leakage detection system} \cite{Prodon2010} etc.\ can be modelled as cases of the PCST problem.

Problems such as Minimum Spanning Tree \cite{GHS_1983, Faloutsos:1995:ODA:224964.225474}, Steiner Tree  \cite{GENHUEY199373, VPK_JCP_GCP_1993, FB_AV_1996, PC_JF_2005}, Steiner Forest  \cite{MK_FK_DM_GP_KT_2008, CL_BPS_2014} etc.\ have been widely studied in the distributed setting. However, such a study
has not been done so far for PCST (the only attempt seems to be a manuscript~\cite{Rossetti_2015}), despite the potential applicability of the problem. In particular, distributed algorithms for PCST would be necessary for solving the problem in distributed \textsl{ad hoc} networks, where nodes may have very limited knowledge of the network.


\smallskip
\noindent
{\bf Our contribution.} In this work we propose a deterministic distributed algorithm for constructing a PCST for a given graph with the set of vertices with their prizes and the set of edges with their weights. Our algorithm is an adaptation of the centralized algorithm proposed by Goemans and Williamson (GW-algorithm) \cite{GW_1995} to distributed setting. In distributed setting each node is a computing entity and can communicate with its neighbors only. Each node locally decides whether it belongs to the Steiner  part or to the Penalty part. The distributed algorithm, we propose, uses primal-dual technique to construct a PCST with an approximation factor of $(2 - \frac{1}{n-1})$ (where $n = |V|$) which is same as that of the Goemans and Williamson's algorithm \cite{GW_1995}. Also it incurs $O(|E||V|)$ message and time complexity. Moreover, the worst case time complexity can be fine tuned to $O(\mathcal{D}|E|)$ where $\mathcal{D}$ is the diameter of the network. Our algorithm uses a novel idea of preserving dual constraints in a distributed way in order to achieve the desired approximation factor. We believe that this technique can prove useful in other tree problems which can be solved using primal-dual method. The main challenge here is to satisfy the dual constraints using local information instead of global information. To this end we use a careful merging and deactivation of connected components so that each component always satisfy its dual constraints. 

One can design a naive distributed algorithm for the PCST problem by `black-box' use of the GW-algorithm (see Subsubsection~\ref{subsec:message_complexity}, also \cite{Rossetti_2015}). Compared to such a naive algorithm D-PCST has slightly larger worst-case complexity, $O(|V||E|)$ vs.\ $O(|V|^2 \log |V|)$. 
However, the time complexity of the naive distributed algorithm is dictated by the complexity of GW-algorithm, hence it holds irrespective of whether the input graph is sparse or dense. On the other hand, for sparse graphs ($|E| = O(|V|)$) our algorithm requires $O(|V|^2)$ time in the worst case (and even better for small diameter sparse graphs) which is a clear improvement compared to the naive approach. Moreover our algorithm, being genuinely distributed, can be adapted to the dynamic setting (node or link additions or deletions) with very low incremental complexity while the naive algorithm would have to run from scratch. Finally, we expect that our techniques can find further applications in obtaining distributed versions of primal-dual based algorithms for tree problems in graphs.



\smallskip
\noindent
{\bf Paper organization.} Section~\ref{related-work} contains the works
related to the PCST problem. In section~\ref{model}, we introduce the formulation of the PCST problem using integer programming (IP) and linear programming (LP). High level description of our distributed PCST (D-PCST) algorithm as well as an illustrating example are given in Section~\ref{description-D-PCST}. The detailed proof of correctness of our distributed PCST algorithm  is given in section~\ref{proof-of-correctness}. Section~\ref{conclusion} contains brief discussion and open questions. The description of the centralized PCST algorithm proposed by Goemans and Williamson \cite{GW_1995} and the pseudo-code of D-PCST are provided in the appendix. 
%

\smallskip
\section{Related Work}\label{related-work}
 The first centralized approximation algorithm for PCST was given by Bienstock et al. \cite{DB_MG_DS_DW_1993} in 1993, although a related problem named prize collecting travelling salesman problem (PCTSP) was introduced earlier by Balas \cite{BALAS_1989}. Bienstock et al. achieved an approximation factor of 3 by using linear programming (LP) relaxation technique. Two years later, based on the work of Agrawal, Klein and Ravi \cite{AA_PK_RR_1995}, Goemans and Williamson \cite{GW_1995} proposed a primal-dual algorithm using the  LP relaxation which runs in $O(n^2 \log n)$ time. The algorithm proposed by Goemans and Williamson consists of two phases namely {\em growth} phase and {\em pruning} phase and yields a solution of approximation factor $(2 - \frac{1}{n-1})$ of the optimal. This algorithm is often denoted as GW-algorithm.  

Johnson et al. \cite{DJ_MM_SP_2000} proposed an improved version of the GW-algorithm maintaining the same approximation factor $(2 - \frac{1}{n-1})$ as of the GW-algorithm. The improvement is achieved by enhancing the \textit{pruning} phase of GW-algorithm which is termed as \textit{strong pruning}.  Johnson et al. also presented a review of different PCST related problems. They modified the growth phase of the GW-algorithm so that it works without a root node.

However, the result of Johnson et al. \cite{DJ_MM_SP_2000} was shown to be incorrect by Feofiloff, Fernandes, Ferreira, and De Pina \cite{PF_CF_CF_JP_2007}. They proved it by a counter example where the algorithm proposed by Johnson et al. returns an approximation factor of 2  instead of $(2 - \frac{1}{n-1})$. They introduced a new algorithm for the PCST problem  based on the GW-algorithm having  a different LP formulation. They achieved a solution of $(2 - \frac{2}{n})$ approximation factor for the unrooted version of the PCST whose running time is $O(n^2 \log n)$. Archer et al. \cite{AA_MB_MH_2011} provided a $(2-\epsilon)$-approximation ($\epsilon > 0$) algorithm for the PCST problem. Specifically the approximation ratio of this algorithm for PCST is below 1.9672.
They achieved this by using the improved Steiner Tree algorithm of Byrka  et al. \cite{Byrka:2010:ILA:1806689.1806769} as a black box in their algorithm.

The ``quota'' version of the PCST problem was studied by Haouari et al. \cite{Haouari_2010}  in which the goal is to find a subtree that includes the root node and 
has a total prize not smaller than the specified quota, while minimizing the cost of the
PCST. A polynomial time algorithm for PCST  was given by  Miranda et al. \cite{EM_AC_XC_XH_BL_2010} for a special network called \textit{2-tree} where prizes (node weights) and edge weights belong to a given interval. This result is based on the work of Wald and Colbourn \cite{Wald_Colbourn_1983} who proved that Steiner Tree problem is polynomial time solvable on 2-tree. An algorithm for Robust Prize-Collecting Steiner Tree problem was proposed by Miranda et al. \cite{EM_IL_PT_2013}. There are other approaches to solve the PCST problem. Canuto et al. \cite{CAUNTO_2001} gave a multi-start local search based algorithm for the PCST problem. Klau et al. \cite{Klau2004} provided an evolutionary algorithm for the PCST problem. All of these are centralized algorithms for the PCST problem. 

The MST and ST problems have been extensively studied in both centralized and distributed setting.  Galleger, Humblet, and Spira \cite{GHS_1983} presented the first distributed algorithm for MST problem with message complexity $O(|E| + |V|\log|V|)$ and time complexity $O(|V|\log |V|)$. 
Later Faloutsos et al. \cite{Faloutsos:1995:ODA:224964.225474} presented a distributed algorithm for MST with message and time complexity $O(|E| + |V|\log|V|)$ and $O(|V|)$ respectively.
Similarly, in the recent years, many distributed algorithms have been proposed for ST and related problems \cite{GENHUEY199373, VPK_JCP_GCP_1993, FB_AV_1996, PC_JF_2005, MK_FK_DM_GP_KT_2008, CL_BPS_2014}. The first heuristic based distributed algorithm for the ST problem in an asynchronous network was proposed by Chen et al. \cite{GENHUEY199373} with the approximation ratio $2(1 - \frac{1}{l})$ of the optimal, where $l$ is the number of leaves in the optimal ST. It has message complexity $O(|E| + |V|(|V \setminus Z| + \log|V|))$ and time complexity $O(|V|(|V \setminus Z|))$ where $Z$ is the set of terminal nodes. Chalermsook et al. \cite{PC_JF_2005} presented a 2-approximation distributed algorithm for the ST problem with time complexity $O(|V| \log |V|)$ on synchronous networks. 
Similarly, there exist distributed algorithms for other variants of the ST problems. 
The first distributed algorithm for Steiner Forest (SF) problem with approximation factor $O(\log |V|)$ was presented by Khan et al. \cite{MK_FK_DM_GP_KT_2008}.  Recently Lenzen et al. \cite{CL_BPS_2014} proposed a distributed algorithm for constructing a SF in {\em congest} model with an approximation factor of $(2 + \epsilon)$ for $\epsilon > 0$ and time complexity of $O(sk + \sqrt{min (st, n)})$  where $s$ is the shortest path diameter, $t$ is the number of terminals, and $k$ is the number of terminal components in the input.
 
To the best of our knowledge our algorithm is the first distributed algorithm for PCST with a constant approximation ratio. Regarding distributed PCST we were able to find only one manuscript in the literature, by Rossetti \cite{Rossetti_2015}, where two algorithms were proposed: The first one is based on MST heuristic and fails to guarantee any constant approximation ratio. The second is a straightforward implementation where GW-algorithm is used as a `black-box' (similar to the naive approach discussed in Subsubsection~\ref{subsec:message_complexity}); as mentioned in \cite{Rossetti_2015} that algorithm is in essence centralized and of very limited practical value.

\vspace{-1em} 
\section{Model and problem formulation}\label{model}
\vspace{-.8em}
We model the distributed PCST problem on a connected network as a graph $G = (V, E, p, w)$, where vertex set $V$ and edge set $E$ represent the set of nodes and the set of communication links of the network respectively. Each edge $e \in E$ has a non-negative cost denoted by $w_e$.  Each vertex $v \in V$  has an unique identification number and  a non-negative \textit{prize} value denoted by $p_v$. 
We assume that each node in the network knows its own prize value and cost of each of its incident links. 
Each node performs the same local algorithm and communicates and coordinates their actions with their neighbors by passing messages only. We consider that communication links are reliable and messages are delivered in FIFO order. A message sent by a sender is eventually 
received by a receiver. However no upper bound of message delay is assumed. A special node of the network designated as \textit{root} ($r$)  initiates the algorithm. In this work we assume that nodes and links do not fail.

The PCST problem can be formulated as the following integer program (IP).
 \vspace{-.6em}
\begin{gather*}
\vspace{-4em}
Min\ \ \ \ \ 	\sum_{e \in E}w_e x_e + \sum_{U \subset V; r \notin U} z_U \Big(\sum_{v \in U}p_v\Big)\\
\hspace{-26em} Subject \hspace{.4em}to:\\
  \hspace{6.5em}  x(\delta(S)) + \sum_{U \supseteq S} z_U \geq 1 \hspace{8.5em}   S \subset V; r \notin S\\
\hspace{-5em} 	\sum_{U \subset V; r \notin U} z_U \leq 1\\ 
\hspace{6em} 	x_e \in \{0,1\} 			\hspace{10em} 				e \in E\\
\hspace{9em}    z_U \in \{0,1\} 			\hspace{10em} 	 U \subset V; r \notin U
			  	 \vspace{-2em}
\end{gather*}		
For each edge $e\in E$ there is a variable $x_e$ that takes a value in $\{0,1\}$. Here $\delta(S)$ denotes the set of edges having exactly one endpoint in $S$ and 
$x(\delta(S)) = \sum_{e \in \delta(S)} x_e$. For every possible $U \subset V : r\notin U$, there is a variable $z_{U}$ that takes values from $\{0,1\}$. A tree $T = (V', E')$ rooted at the root node $r$ corresponds to the following integral solution of the IP: $x_e = 1$ for each $e \in E'$, $z_{V \setminus V'} = 1$ and all other variables are zero. The first integral constraint says that a subset of nodes $S \subset V$ ($r \notin S$) is connected to $T$ if there exists at least one $e \in \delta(S)$ such that $x_e = 1$ or it is not connected to $T$ if $S \subseteq U \subset V$ ($r \notin U$), $x_e =0$ for all $e \in \delta(S)$ and $z_U =1$. 
The second integral constraint of the IP implies that there can be at most one such $U \subset V$ such that $r \notin U$ for which $z_{U}=1$. Note that we can set $p_r = \infty$ since every feasible tree is required to include the root node $r$.

Since finding the exact solution of an IP is NP-hard and LP (linear programming) is polynomial time solvable, therefore we generally go for its LP-relaxation and find an approximate solution for the problem. 
Note that dropping of the constraint $\sum_{U \subset V; r \notin U} z_U \leq 1$ from LP-relaxation does not affect the optimal solution, therefore we exclude it from the LP-relaxation. The corresponding LP-relaxation is as follows:
\begin{gather*}
\vspace{-2em}
Min\ \ 	\sum_{e \in E}w_e x_e + \sum_{U \subseteq V; r \notin U} z_U \Big(\sum_{v \in U}p_v\Big)\\
\hspace{-28em} Subject \hspace{.4em}to:\\
   \hspace{5.5em}  \sum_{e\in \delta(S)} x_e + \sum_{U \supseteq S} z_U \geq 1 \hspace{8em}   S \subset V; r \notin S\\
			\hspace{7.5em} 	x_e \geq 0		\hspace{10em} 				 e \in E\\
			  	 \hspace{11em}    z_U \geq 0  \hspace{10em} 	 U \subset V; r \notin U
			  	 \vspace{-4em}
\end{gather*}

\vspace{-.2em}
The above LP-relaxation has  two types of basic variables namely $x_e$ and $z_U$ and exponential number of constraints.  If it is converted into its dual then there will be one type of basic variables and two types of constraints. Also by weak LP-duality every feasible solution to the dual LP gives a lower bound on the optimal value of the primal LP. The dual of the above LP-relaxation is as follows:
\begin{gather*}
\hspace{-4.5em}Max \sum_{S\subseteq V -\{r\}}  y_S\\
\hspace{-26em} Subject \hspace{.4em}to:\\
 \hspace{6em} \sum_{S:e \in \delta(S)}y_S \leq w_e  	\hspace{9.6em} 		e \in E\\
		\hspace{10em} \sum_{S \subseteq U} y_S \leq \sum_{v \in U}p_v		\hspace{9em} U \subset V; r \notin U \\
			\hspace{10em}	y_S	\geq 0		\hspace{13em} 		S \subset V; r \notin S 
\end{gather*}		
Here the variable $y_{S}$  corresponds to the primal constraint $\sum_{e\in \delta(S)} x_e + \sum_{U \supseteq S} z_U \geq 1 $.	 
The dual objective function indicates that for each $S \subseteq V \setminus \{r\}$,  the variable $y_S$ can be increased as much as possible without violating the two dual constraints $\sum_{S:e \in \delta(S)}y_S \leq w_e$  and $\sum_{S \subseteq U} y_S \leq \sum_{v \in U}p_v$.
The constraint $\sum_{S:e \in \delta(S)}y_S \leq w_e$ is known as {\em edge packing} constraint which is corresponding to the primal variable $x_e$. It says that for each $S \subseteq V \setminus \{r\}$ such that $e \in \delta(S)$, $y_S$ can be increased as much as possible until the {\em edge packing} constraint becomes tight,  i.e. $\sum_{S:e \in \delta(S)}y_S = w_e$. This equality implies the case where the primal variable $x_e=1$ for the corresponding edge $e$, and $e$ is added to the forest being constructed. The value $w_e$ contributes to the primal objective value of the PCST.  The dual constraint $\sum_{S \subseteq U} y_S \leq \sum_{v \in U}p_v$ is known as {\em penalty packing} constraints which is corresponding to the primal variable $z_U$ of the LP relaxation. For each $S \subseteq U$ such that $r \notin U$, $y_S$ can be increased as much as possible until the penalty packing constraint becomes tight i.e. $\sum_{S \subseteq U} y_S = \sum_{v \in U}p_v$. 
Any positive value of $y_S$ can be considered feasible provided it does not lead to the violation of any of the two dual packing constraints. If we set $y_S=0$ for each $S \subseteq V \setminus {r}$ then it gives a trivial feasible solution to the dual LP since it satisfies both the packing constraints. The dual LP is feasible at its origin ($y_S = 0$ for each $S \subseteq V \setminus {r}$), whereas primal LP is not feasible at its origin ($x_e = 0$ for each $e \in E$ and $z_U = 0$ for each $U \subseteq V \setminus {r}$).

\section{Description of the D-PCST algorithm} \label{description-D-PCST}

{\bf Terminology}. 
A set of nodes $C \subseteq V$ connected by a set of edges is termed as {\em component}. Each component has a state which can be $\mathit{sleeping}$, $\mathit{active}$ or $\mathit{inactive}$. At node $v \in V$, the state of an incident edge $e$ is denoted as $\mathit{SE(e)}$. The value of $\mathit{SE(e)}$ can be $\mathit{basic}$, $\mathit{branch}$, $\mathit{rejected}$ or $\mathit{refind}$.  At node $v$ the state of an edge $e \in \delta(v)$ is $\mathit{branch}$ if $e$ is selected as a candidate branch edge for the Steiner tree of the PCST.  Any edge inside a component (between two nodes $u,v \in C$) which is not a $\mathit{branch}$ is stated as $\mathit{rejected}$.  If node $v$ receives a message $\mathit{refind\_epsilon}$ on $e \in \delta(v)$ then $\mathit{SE(e)=refind}$. An edge $e$ which  is neither  $\mathit{branch}$ nor $\mathit{rejected}$ nor $\mathit{refind}$ has the state named $\mathit{basic}$. Each component has a leader node which coordinates all the activities inside the component. Each node $v \in \mathit{C}$ locally knows the current state of its component $C$, denoted as $CS(C)$. Each component $C$ has a weight denoted by $W(C)$ which is known to each $v\in C$.
Each node $v$ has a $\mathit{deficit}$ value denoted by $d_v$. The constraint $d_v + d_u \leq w_e$ always holds for any edge $e = (v, u) \in E$. In addition, the following symbols and terms are used in the description of our algorithm.
\begin{itemize}
	\setlength\itemsep{.1em}
	\item $\mathit{\delta(v)}$ denotes the set of edges incident on $v$.
	\item $\mathit{\epsilon_e}$ is a value calculated for an edge $e$.  
	\item $\mathit{\epsilon_1(v)} = \smash{\displaystyle\min_{e \in \mathit{\delta(v)} \cap \mathit{\delta(C)}}} \{\mathit{\epsilon_e}\}$.
	\vspace{.3em}
	\item $\mathit{\epsilon_1(C) = \smash{\displaystyle\min_{v \in C}} \hspace{.5em} \{ \epsilon_1(v)\}}$.
	\vspace{.3em}
	\item $\mathit{\epsilon_2(C) = \sum_{v \in C} p_v - W(C)}$.
	\item $\mathit{d_h(C) = \smash{\displaystyle\max_{ v \in C}}  \hspace{.5em} \{d_v\}}$. We use $\mathit{d_h(C)}$ to denote the highest {\em deficit} value of a component $C$.
		\vspace{.2em}
	\item MOE ({\em minimum outgoing  edge}) is the edge $e \in \mathit{\delta(C)}$ that gives the $\mathit{\epsilon_1(C)}$. 
	\item  A component $C'$ is  called a {\em neighboring} component of a component $C$ if $\mathit{\delta(C) \cap \delta(C') \neq \phi}$.
	\item $\langle M \rangle$ denotes the message $M(a_1, a_2,...)$. Here $a_1, a_2,...$ are the arguments of message $M$. Note that unless it is necessary arguments of $\langle M \rangle$ will not be shown in it.  
\end{itemize}

\noindent
{\bf Input and output specification}. 
Initially each node $v \in V$ knows its own prize value $p_v$, unique identity, and weight $w_e$ of each edge $e \in \mathit{\delta(v)}$. Before the start of the algorithm, $\mathit{prize\_flag = TRUE}$ for all $v \neq r$. If $v=r$ then $\mathit{prize\_flag = FALSE}$. Also each node $v \in V$ initially sets its local boolean variable $\mathit{labelled\_flag = FALSE}$ and  $\mathit{SE(e) = basic}$ for each $\mathit{e \in \delta(v)}$. When the D-PCST algorithm terminates, each node $v \in V$ knows whether it is in the Penalty part or in the Steiner part. A node $v$ belongs to the Penalty part if its local variable $\mathit{prize\_flag}$ is set to $\mathit{TRUE}$. Otherwise it belongs to the Steiner part. In addition, if a node $v$ belongs to the Steiner part then at least one $\mathit{e \in \delta(v)}$ must be assigned as a $\mathit{branch}$ edge.  So the pair $\mathit{(prize\_flag, Y)}$ at each node $v$ clearly defines the distributed output of the algorithm. Here $\mathit{Y \subseteq \delta(v)}$. If $\mathit{prize\_flag = TRUE}$ then $Y = \phi$. Otherwise for each $e\in Y$, $\mathit{SE(e)=branch}$.  

\medskip
\noindent
{\bf Basic principle}. 
Our algorithm consists of two phases namely {\em growth} phase and {\em pruning} phase. 
At the beginning of the algorithm each component comprises of a single node. Initially each component except the root component (containing the special node $r$, the root) is in $\mathit{sleeping}$ state. The initial state of the root component is $\mathit{inactive}$. The root node $r$ initiates the algorithm. At any instant of time  the algorithm maintains a set of components. The growth phase performs the following operations until all components in the network become $\mathit{inactive}$. 
\begin{enumerate}[(i)]
\setlength\itemsep{-.1em}
	\item Merging: merging of two distinct neighbouring components $\mathit{C}$ and  $\mathit{C'}$.
	\item Deactivation: an active component becomes inactive.
	\item Proceed: an inactive component $\mathit{C}$ sends $\langle proceed \rangle$ to a neighboring component $\mathit{C'}$.
	\item Back: an inactive component $\mathit{C}$ sends $\langle back \rangle$ to an inactive neighboring component $\mathit{C'}$.
\end{enumerate}
Whenever all components in the network become inactive, the growth phase terminates and the pruning phase starts. The pruning phase prunes a component or subcomponent from the structure constructed in the growth phase if such a pruning enhances the PCST.

\medskip

\noindent
{\bf Growth phase}. Initially, $\mathit{d_v} =0$ and $\mathit{W(C)} = 0$ at each node $v \in V$ ($v \in \mathit{C}$).  
At any point of time only one component $\mathit{C}$ calculates its  $\mathit{\epsilon(C)}$. The leader of $\mathit{C}$ computes  $\mathit{\epsilon(C) = \min (\epsilon_1(C), \epsilon_2(C))}$ using message passing. 
Depending on the value of $\mathit{\epsilon(C)}$, the leader of $\mathit{C}$ proceeds with any one of the  following actions.
\begin{enumerate}[(i)]
\setlength\itemsep{-.1em}
	\item If $\mathit{CS(C) = active}$ then it may decide to merge with one of its neighboring component $\mathit{C'}$ or it may decide to become inactive.
	\item If $\mathit{CS(C) = inactive}$ then it asks one of its neighboring components, say $\mathit{C'}$,
	to proceed further. The choice of $\mathit{C'}$ depends on the value $\mathit{\epsilon_1(C)}$ computed at $\mathit{C}$.  Note that an inactive component $\mathit{C}$ never computes  $\mathit{\epsilon_2(C)}$ and its $\mathit{\epsilon(C)}$ is equal to $\mathit{\epsilon_1(C)}$.
\end{enumerate}

To compute $\mathit{\epsilon_1(C)}$, the leader of $\mathit{C}$ broadcasts  $\langle initiate \rangle$ using the tree rooted at the leader (set of branch edges inside a component $\mathit{C}$ forms a tree rooted at the leader) asking each {\em frontier} node $v \in \mathit{C}$ to finds its $\mathit{\epsilon_1(v)}$. A node $v \in \mathit{C}$ is called a {\em frontier} node if it has at least one edge $\mathit{e \in \delta(v) \cap \delta(C)}$. Upon receiving $\langle initiate \rangle$, each frontier node $v \in \mathit{C}$ calculates $\mathit{\epsilon_e}$ for each edge $\mathit{e \in \delta(v) \cap \delta(C)}$. Note that if an edge $e$ satisfies the condition $\mathit{e \in \delta(v) \cap \delta(C)}$ then the state of the edge $e$ at node $v$ is either $\mathit{basic}$ or $\mathit{refind}$. Let $\mathit{C'}$ be a neighboring component of $\mathit{C}$ such that $\mathit{e \in \delta(u)}$,  $\mathit{e \in \delta(v)}$,  $\mathit{u \in C'}$, and $\mathit{v \in C}$. Now $v$ calculates $\mathit{\epsilon_e}$ as follows. 

\begin{enumerate}[(i)]
\item $\mathit{CS(C) = active}$ and $\mathit{CS(C') = active}$ : in this case $\mathit{\epsilon_e = \frac{w_e - d_v - d_u}{2}}$.
\item $\mathit{CS(C) = active}$ and $\mathit{CS(C') = inactive}$ : in this case $\mathit{\epsilon_e = w_e - d_v - d_u}$. 
\item $\mathit{CS(C) = active}$ and $\mathit{CS(C') = sleeping}$ : in this case $\mathit{\epsilon_e = \frac{w_e - d_v - d_u}{2}}$. Here the state of the component $\mathit{C'}$ is $\mathit{sleeping}$ and therefore the deficit value $\mathit{d_u}$ of the node $\mathit{u \in C'}$ is considered to be equal to $\mathit{d_h(C)}$.
\item $\mathit{CS(C) = inactive}$ and $\mathit{CS(C') = sleeping}$ : in this case $\mathit{\epsilon_e = w_e - d_v - d_u}$. Similar to the previous case the deficit value $\mathit{d_u}$ is equal to $\mathit{d_h(C)}$.
\item $\mathit{CS(C) = inactive}$ and $\mathit{CS(C') = inactive}$ : in this case the value of $\mathit{\epsilon_e}$ for an edge $\mathit{e \in \delta(v)}$ calculated by $v$ depends on the state of the edge $e$. If $\mathit{SE(e) = refind}$ then $\mathit{\epsilon_e = w_e - d_v - d_u}$. Otherwise, $\mathit{\epsilon_e = \infty}$.
\end{enumerate} 
Note the following cases. 
\begin{itemize}
\item  Whenever an inactive component $\mathit{C}$ is in the state of computing its $\mathit{\epsilon_1(C)}$ then there cannot exist any component $\mathit{C'}$ in the  neighborhood of $\mathit{C}$ such that $CS(\mathit{C'})=active$.
\item A component $\mathit{C}$ never computes its $\mathit{\epsilon_1(C)}$ (or $\mathit{\epsilon_2(C)}$) while it is in the  $\mathit{sleeping}$ state. 
\end{itemize}
\noindent
Following these conditions a frontier node $v$ calculates the value of $\mathit{\epsilon_e}$ for each of its incident $\mathit{basic}$ or $\mathit{refind}$ edge and $\mathit{\epsilon_1(v)}$ is locally selected for reporting to the leader of $\mathit{C}$. In this way each frontier node $\mathit{v \in C}$ locally calculates $\mathit{\epsilon_1(v)}$  and reports it to the leader using convergecast technique over the tree rooted at the leader of the component. During the convergecast process each node sends its $\langle report \rangle$ which contains $\mathit{\epsilon_1(v)}$ to its parent node using its incident branch edge of the tree. In this process overall $\mathit{\epsilon_1(C)}$ survives and eventually reaches the leader node of $\mathit{C}$. Also during the convergecast  each node $\mathit{v \in C}$ reports the total prize value of all the nodes in the subtree rooted at $v$. Therefore, eventually the total prize of the component $\mathit{C}$ ($\mathit{TP(C)}$) is also known to the leader. If $\mathit{CS(C) = active}$ then the leader of $\mathit{C}$ calculates $\mathit{\epsilon_2(C) = \sum_{v\in C}p_v - W(C) =  TP(C) - W(C)}$ .

The leader of $\mathit{C}$ now computes $\mathit{\epsilon(C) = \min(\epsilon_1(C), \epsilon_2(C))}$.
If $\mathit{\epsilon(C) = \epsilon_2(C)}$ then $\mathit{C}$ decides to deactivate itself.  This indicates that the dual penalty packing constraint $\mathit{\sum_{S \subseteq C} y_S \leq \sum_{v \in C}p_v}$ becomes tight for the component $\mathit{C}$. On the other hand, if $\mathit{\epsilon(C) = \epsilon_1(C)}$ then it indicates that for component $\mathit{C}$ the dual edge packing constraint $\mathit{\sum_{S:e \in \delta(S)}y_S \leq w_e}$ becomes tight for one of the edges $e$ such that $e \in \delta(\mathit{C}) \wedge e \in \delta(\mathit{C}')$, where $\mathit{C'}$ is neighboring component of $\mathit{C}$. In this case the leader of $\mathit{C}$ sends a $\langle merge(\mathit{\epsilon(C), d_h(C)})$ to a frontier node $\mathit{v \in C}$ which resulted the $\mathit{\epsilon(C) = \epsilon_1(C)=\epsilon_1(v)}$. Upon receiving $\langle merge(\mathit{\epsilon(C), d_h(C)})$ the node $v$ sends $\langle \mathit{connect(v, W(C), d_v, d_h(C))} \rangle$ over the MOE to $\mathit{C'}$ to  merge with it. Whenever a node $\mathit{u \in C'}$ receives $\langle \mathit{connect(v, W(C), d_v, d_h(C))} \rangle$ on an edge $\mathit{e \in \delta(u)}$ then depending on the state of $\mathit{C'}$ following actions are taken.
\begin{enumerate}[(i)]
	\item  $\mathit{CS(C') = inactive}$ : in this case the node $\mathit{u \in C'}$ sends $\langle accept \rangle$ to $\mathit{v\in C}$. This confirms the merging of two components $\mathit{C'}$ and $\mathit{C}$.
	\item $\mathit{CS(C') = sleeping}$ : in this case it is obvious that $\mathit{C'}$ is a single node component  $\{u\}$. The state of $\mathit{C'}$ becomes $\mathit{active}$ and each of its local variables $d_u$, $\mathit{W(C')}$ and $\mathit{d_h(C')}$ is initialized to $\mathit{d_h(C)}$ (received in the $\langle connect \rangle$). After that the leader of the component $\mathit{C'}$ ($u$ itself) computes $\mathit{\epsilon_e}$ for edge $e$ and $\mathit{\epsilon_2(C')}$.  If $\mathit{\epsilon_e < \epsilon_2(C')}$ then the component $\mathit{C'}$ sends $\langle accept \rangle$ to the component $\mathit{C}$ which confirms the merging of two components $\mathit{C'}$ and $\mathit{C}$. On the other hand if $\mathit{\epsilon_2(C') \leq \epsilon_e}$  then $\mathit{C'}$ decides to deactivate itself and sends $\langle refind\_epsilon \rangle$ to the component $\mathit{C}$.  
\end{enumerate}	

Whenever a node $\mathit{v \in C}$ receives $\langle \mathit{refind\_epsilon} \rangle$ in response to a  $\langle \mathit{connect}\rangle$ on an edge $e \in \delta(v)$ then the state of its local variable $\mathit{SE(e)}$ becomes $\mathit{refind}$. The node $v$ also reports the $\langle \mathit{refind\_epsilon}\rangle$ to the leader of $\mathit{C}$. Upon receiving $\langle \mathit{refind\_epsilon} \rangle$ the leader node of $\mathit{C}$ proceeds to calculate its $\mathit{\epsilon(C)}$ once again. 

Whenever a component $\mathit{C}$ decides to merge or deactivate (only an active component can decide to deactivate itself) then each node $\mathit{v \in C}$ increases each of  $d_v$ and $\mathit{W(C)}$ by  $\mathit{\epsilon(C)}$, and $\mathit{d_h(C)}$ is also updated. Note that for each component $\mathit{C}$ there is an implicit  dual variable $\mathit{y_{C}}$ which we want to maximize subject to the dual constraints. Whenever the local variables of a component $\mathit{C}$ are updated by $\mathit{\epsilon(C)}$, $\mathit{y_C}$ is also implicitly updated. 

If two components $\mathit{C}$ and $\mathit{C'}$ decide to merge through an edge $\mathit{e = (v, u): v \in C \wedge u \in C'}$  then the dual {\em edge packing} constraint $\mathit{\sum_{S:e \in \delta(S)}y_S \leq w_e}$ becomes tight for the edge $e$. Both nodes $v$ and $u$ set their local variables $\mathit{SE(e) = branch}$ for edge $e$. The weight of the resulting component $\mathit{C \cup C'}$ is the sum of the weights of  $\mathit{C}$ and $\mathit{C'}$, i.e. $\mathit{W(C \cup C') = W(C) + W(C')}$. If $\mathit{C \cup C'}$ contains the root node $r$ then it becomes $\mathit{inactive}$ (root component is always $\mathit{inactive}$) and $r$ remains the leader of the new component $\mathit{C \cup C'}$. In addition, whenever a component $\mathit{C'}$ merges with the root component then each $\mathit{v \in C'}$ sets its local variable $\mathit{prize\_flag = FALSE}$ and there exists at least one edge $ \mathit{e \in \delta(v)}$ such that $\mathit{SE(e) = branch}$. This indicates that each node $\mathit{v \in C'}$ contributes to the Steiner part of the  PCST. On the other hand if none of the merging components $\mathit{C}$ or $\mathit{C'}$ is the root component then the resulting component $\mathit{C \cup C'}$ becomes $\mathit{active}$. In this case the node with the higher ID between the two adjacent nodes of the merging edge becomes the new leader of  $\mathit{C \cup C'}$ and for each node $\mathit{v \in C \cup C'}$ the boolean variable $\mathit{prize\_flag}$ remains $\mathit{TRUE}$. 

In case of deactivation of a component $\mathit{C}$, each node $\mathit{v \in C}$ sets its $\mathit{labelled\_flag = TRUE}$. Whenever an active component $\mathit{C}$ becomes inactive and there exists no active component in its neighborhood then the leader of $\mathit{C}$ may decide to send $\langle proceed(\mathit{d_h(C)}) \rangle$ or $\langle back \rangle$ to one of its neighboring component $\mathit{C'}$. For this, first of all the leader of $\mathit{C : CS(C) = inactive}$ computes its $\mathit{\epsilon_1(C)}$. Note that in this state of the network there cannot exist any component $\mathit{C'}$ in the  neighborhood of $\mathit{C}$ such that $\mathit{CS(C')=active}$. The inactive component $\mathit{C}$ finds its $\mathit{\epsilon_1(C)}$ only. The value of $\mathit{\epsilon_1(C)}$ may be some finite real number or $\infty$. If $\mathit{C}$ has  at least one neighboring component  $\mathit{C'}$  such that $\mathit{CS(C') = sleeping}$ or if $\mathit{CS(C') = inactive}$ for each neighboring component $\mathit{C'}$ of $\mathit{C}$ and $\mathit{\exists e :  e \in \delta(C) \wedge SE(e) = refind}$ then the value of $\mathit{\epsilon_1(C)}$  is guaranteed to be a finite real number. Otherwise the value of $\mathit{\epsilon_1(C) = \infty}$. If  $\mathit{\epsilon_1(C)}$ corresponding to the edge $\mathit{e : e \in \delta(C) \cap \delta(C')}$ is a finite real number then the component $\mathit{C}$ sends $\langle proceed(\mathit{d_h(C)}) \rangle$ to the component $\mathit{C'}$ through edge $e$. Upon receiving $\langle proceed(\mathit{d_h(C)}) \rangle$, the component $\mathit{C'}$ starts computing its $\mathit{\epsilon(C')}$ for taking further actions. If the value of $\mathit{\epsilon_1(C) = \infty}$ then the leader of the component $\mathit{C}$ sends $\langle back \rangle$ to a neighboring component $\mathit{C''}$ from which it received  $\langle proceed(\mathit{d_h(C'')}) \rangle$ in some early stages of the algorithm. Note that there may be more than one pending $\langle proceed \rangle$ on a component $\mathit{C}$ but a node has only one pending $\langle proceed \rangle$ at a time. In this case the leader of $\mathit{C}$ sends $\langle back \rangle$ to the {\em frontier} node which received the earliest $\langle proceed \rangle$. A {\em frontier} node remember the time at which it receives a $\langle proceed \rangle$ by using a local variable $\mathit{received\_ts}$. It is clear that the number of $\langle back \rangle$ generated by an $\mathit{inactive}$ component depends on the number of pending $\langle proceed \rangle$.  Eventually  when the leader  of the root component $\mathit{C_r}$ finds $\mathit{\epsilon_1(C_r) = \infty}$ then all components in the whole network become $\mathit{inactive}$. This ensures the termination of the {\em growth} phase. After the termination of the growth phase, the root node initiates the {\em pruning} phase. 

\medskip

\indent
{\bf Pruning phase}. In this phase following operations are performed. 
\begin{itemize}
	\item Each node $\mathit{v \in C}$, where $\mathit{C}$ is a non-root inactive component, sets $\mathit{SE(e) = basic}$ for each edge $\mathit{e \in \delta(v)}$ if $\mathit{SE(e) \neq basic}$.
	\item In the root component $\mathit{C_r}$, pruning starts in parallel at each leaf node of the Steiner tree rooted at the root node $r$ and repeatedly applied at every leaf node $v \in C_r$ at any stage as long as the following two conditions hold.
	\vspace{-.6em}
	\begin{enumerate}[(i)]
		\item $\mathit{labelled\_flag = TRUE}$ at node $v$.
		\item There exists exactly one edge $\mathit{e \in \delta(v)}$ such that $\mathit{SE(e) = branch}$. 
	\end{enumerate}	
Each pruned node $\mathit{v \in C_r}$ sets its local variables $\mathit{prize\_flag = TRUE}$, $\mathit{labelled\_flag = FALSE}$, and $\mathit{SE(e) = basic}$ for each edge $e$ such that $\mathit{e \in \delta(v) \wedge SE(e)= branch}$. Note that once a node $v$ changes the value of $\mathit{SE(e)}$ from $\mathit{branch}$ to $\mathit{basic}$ for an edge $e = (v, u) \in \delta(v)$ then the node $u$ also does the same. 
Finally for each of the non-pruned nodes $\mathit{u \in C_r}$, $\mathit{prize\_flag = FALSE}$ and there exists at least one edge $\mathit{e \in \delta(u)}$ such that $\mathit{SE(e) = branch}$.
\end{itemize}

\vspace{-2em} 
\subsection{An Example}\label{example}
\vspace{-.5em}
In this subsection we illustrate the working principle of our proposed D-PCST algorithm with an
example. Due to space constraints, we illustrate only the major operations. 

Figure~\ref{eg-merge} illustrates the merging of two components. Each node has a prize value that is labelled just outside the node. For example, the prize of node $v_1$ is $10$.
Similarly each edge is labelled with an weight. Figure~\ref{eg-merge}(a)  shows the graph before the merging of two neighboring components $\mathit{C=\{v_2, v_5\}}$ which is in $\mathit{active}$ state and $\mathit{C'=\{v_1\}}$ which is in $\mathit{sleeping}$ state. The MOE of the component $\mathit{C}$ is $(v_2, v_1)$ which gives $\mathit{\epsilon_1(C)} = -1$. The leader node $v_5$ also computes $\mathit{\epsilon_2(C)} = 6$. Hence $\mathit{\epsilon(C) = \min(\epsilon_1(C), \epsilon_2(C)) = \epsilon_1(C)}$. So  $v_2$ sends $\langle connect(v_2, 14, 7, 7) \rangle$ over the MOE to merge with $\mathit{C'}$. Now $\mathit{C'}$ becomes $\mathit{active}$ and finds $\mathit{\epsilon(C') = \epsilon_1(C')} = -1$ and $(v_1, v_2)$ is the MOE. Therefore it decides to merge with $C$. The new active component $\{v_1, v_2, v_3\}$ is shown in Figure~\ref{eg-merge}(b). The rectangular box below the graph shows the value of local variables $d_i$ and $W_i$ for each $v_i \in V$.

Figure~\ref{eg-deactivation} shows the deactivation of an active component $\mathit{C=\{v_7, v_{11}\}}$. In Figure~\ref{eg-deactivation}(a) the leader of $\mathit{C}$ finds that its MOE is $(v_7, v_3)$ which gives $\mathit{\epsilon_1(C)} = 7.5$. $\mathit{C}$ also computes its $\mathit{\epsilon_2(C)}$ which is equal to $3$. Since $\mathit{\epsilon(C) = \min(\epsilon_1(C), \epsilon_2(C)) = \epsilon_2(C)}$, therefore the component $\mathit{C}$ deactivates itself. Each node of $\mathit{C}$ sets its local boolean variable $\mathit{labelled\_flag = TRUE}$. The graph after the deactivation of $\mathit{C}$ is shown in Figure \ref{eg-deactivation}(b). 

Figure~\ref{eg-proceed} shows the action of {\em proceed} operation performed by an inactive component $\mathit{C=\{v_3\}}$. In Figure~\ref{eg-proceed}(a), the MOE of $\mathit{C}$ is $(v_3, v_4)$ which gives $\mathit{\epsilon_1(C)} = 10$.  The component $\mathit{C}$ sends $\langle proceed(\mathit{d_h(C)}) \rangle$ (denoted by $P(15)$ in the figure) over its MOE to the component $\mathit{C'=\{v_4\}}$. Upon receiving $P(15)$, the $\mathit{sleeping}$ component $\mathit{C}'$ becomes $\mathit{active}$ and initializes its local variables $d_4$ and $W_4$ to $15$. 
Since $\mathit{\epsilon(C') = \epsilon_1(C')}$, $\mathit{C'}$ sends a connection request to $\mathit{C}$ which is shown in Figure~\ref{eg-proceed}(b).

Figure~\ref{pruning-phase} shows the case of {\em pruning} operation performed in each component of the graph. In Figure~\ref{pruning-phase}(a), inside each of the non-root inactive components the state of each $\mathit{branch}$ edge changes to $\mathit{basic}$. In the root component $\mathit{C_r}$, nodes $v_7$ and $v_{11}$ are pruned. The component $\mathit{C = \{v_7, v_{11}\}}$ was deactivated at some early stage of the {\em growth} phase of the algorithm. At $v_{11}$ the local variable $prize\_flag$ is set to $\mathit{TRUE}$ and $\mathit{SE((v_{11}, v_7))}$ is set to $\mathit{basic}$. Similarly the node $v_7$ and its corresponding adjacent  $\mathit{branch}$ edges are also pruned from the root component. The Figure~\ref{pruning-phase}(b) shows the state of the graph after the {\em pruning} phase which is the final solution to the PCST.



\begin{figure}[!tp]
\vspace{-11em}
\hspace*{-2em}
  \begin{subfigure}[b]{0.5\textwidth}
    \includegraphics[width=10cm]{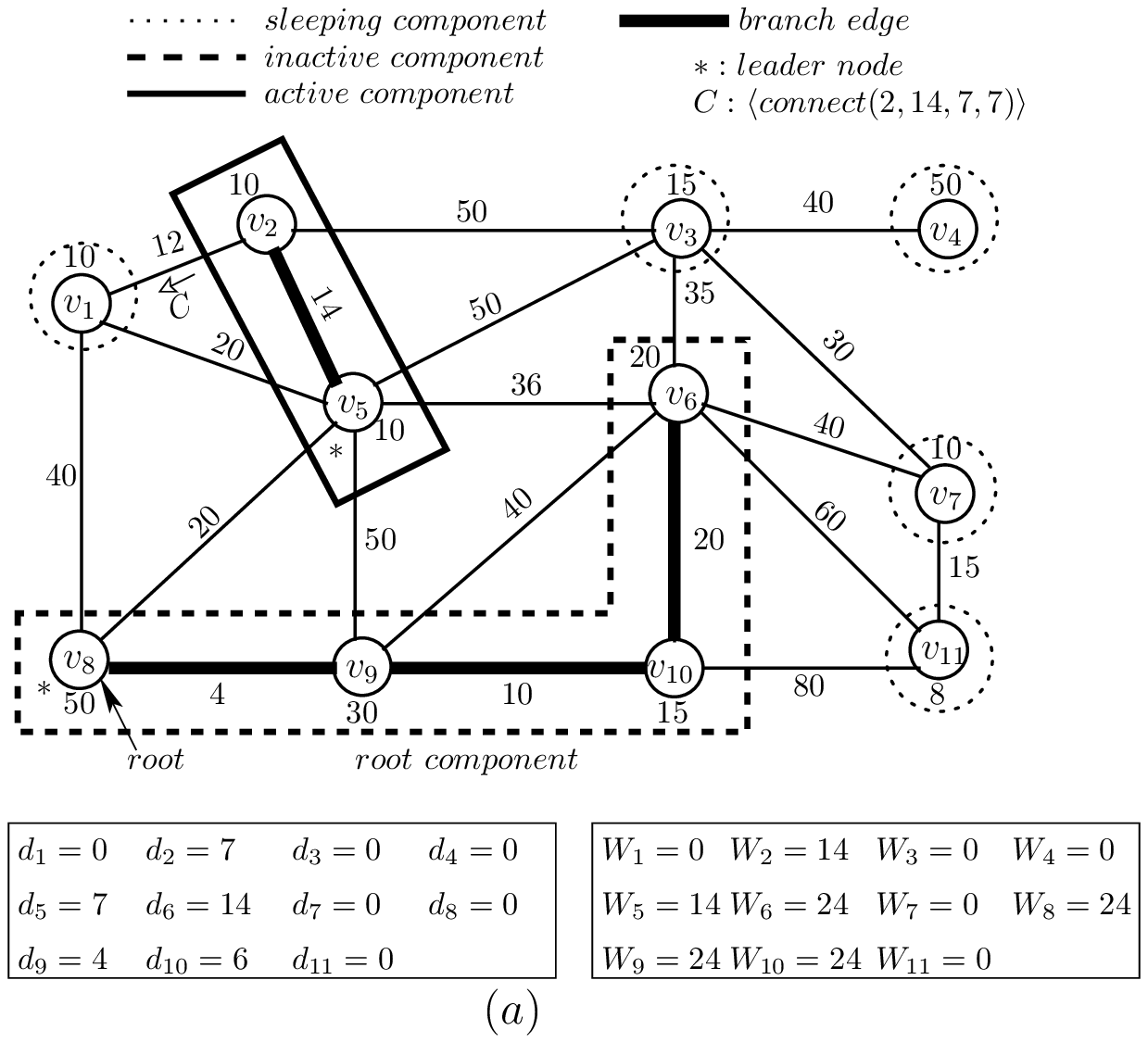}
    \vspace{-10em}
  \end{subfigure}
 \hspace*{-2em}
  \begin{subfigure}[b]{0.5\textwidth}
    \includegraphics[width=10cm]{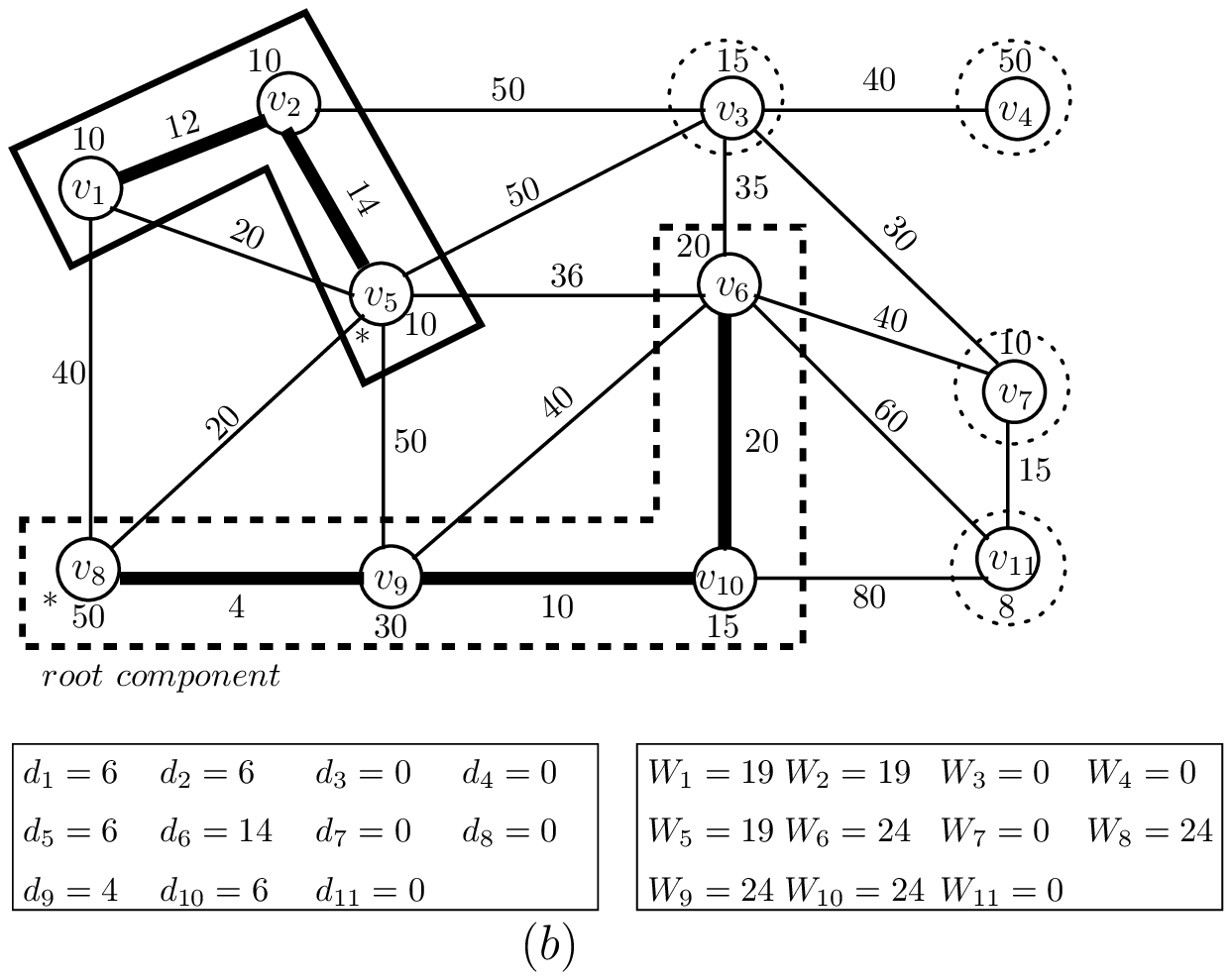}
    \vspace{-10em}
  \end{subfigure}
  \vspace{-2em}
  \captionsetup{font=small}
  \caption{A case of merging operation. (a) state before merging of components $\{v_2, v_5\}$ and $\{v_1\}$. (b) state after merging.} \label{eg-merge}
\end{figure}
\begin{figure}[!tp]
\vspace{-11em}
\hspace*{-2em}
  \begin{subfigure}[b]{0.5\textwidth}
    \includegraphics[width=10cm]{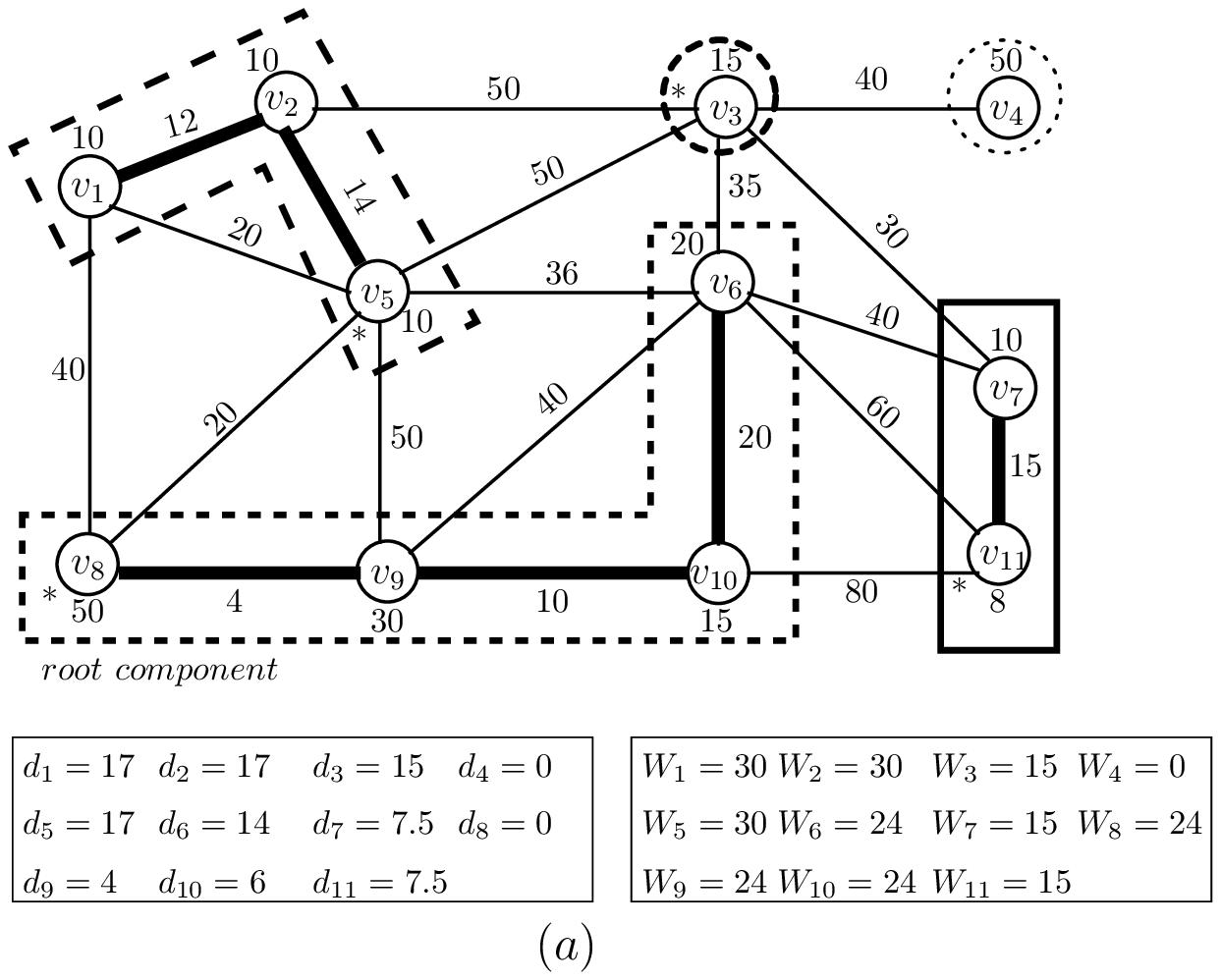}
    \vspace{-10em}
    \label{fig:f1}
  \end{subfigure}
 \hspace*{-3em}
  \begin{subfigure}[b]{0.5\textwidth}
    \includegraphics[width=10cm]{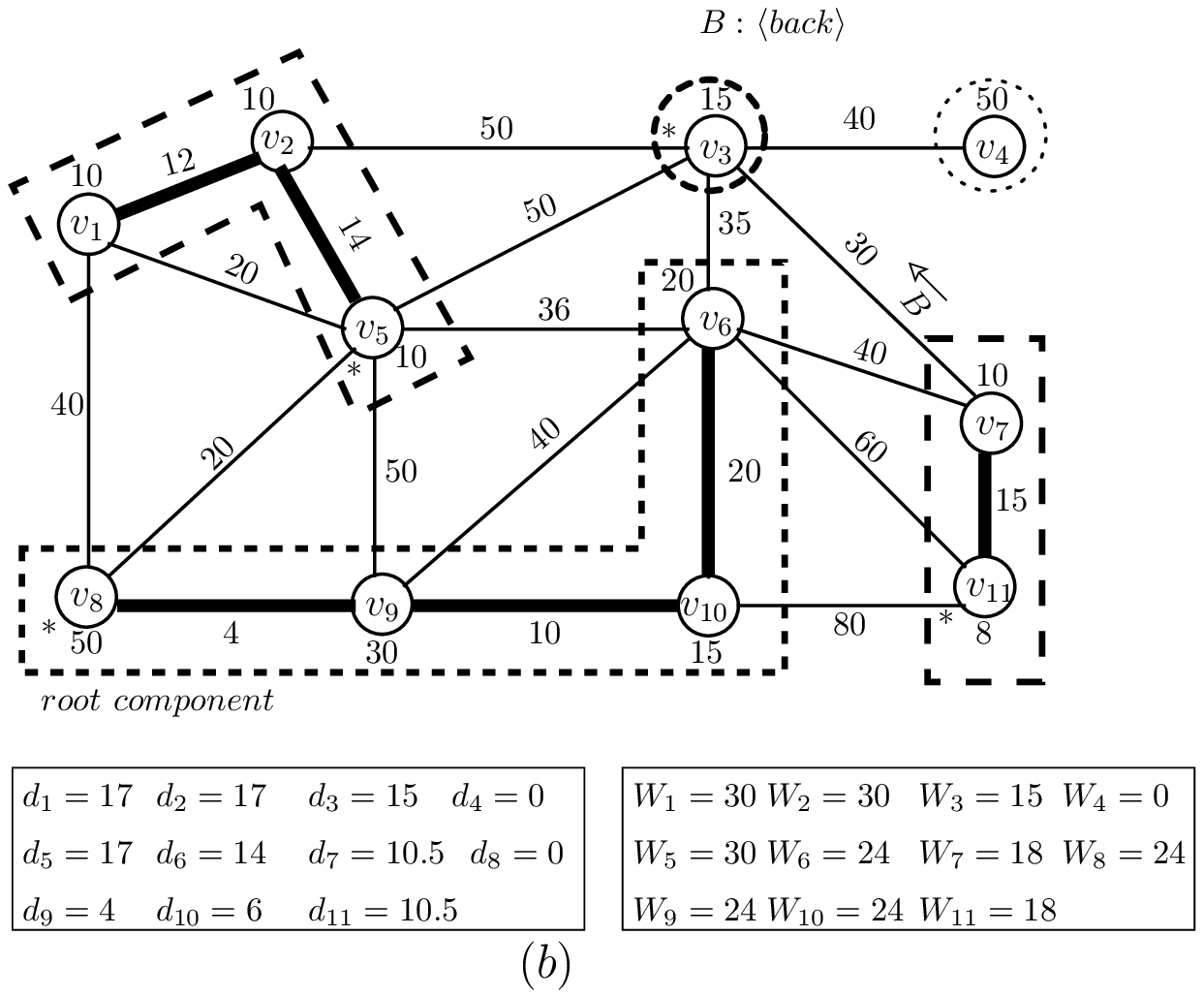}
    \vspace{-10em}
    \label{fig:f2}
  \end{subfigure}
  \vspace{-3em}
  \captionsetup{font=small}
  \caption{A case of deactivation. (a) state  before the deactivation of the active component $\{v_7, v_{11}\}$. (b) state  after the deactivation .} \label{eg-deactivation}
\end{figure}
\begin{figure}[!tp]
\vspace{-18em}
  \begin{subfigure}[b]{0.5\textwidth}
    \includegraphics[width=9.5cm]{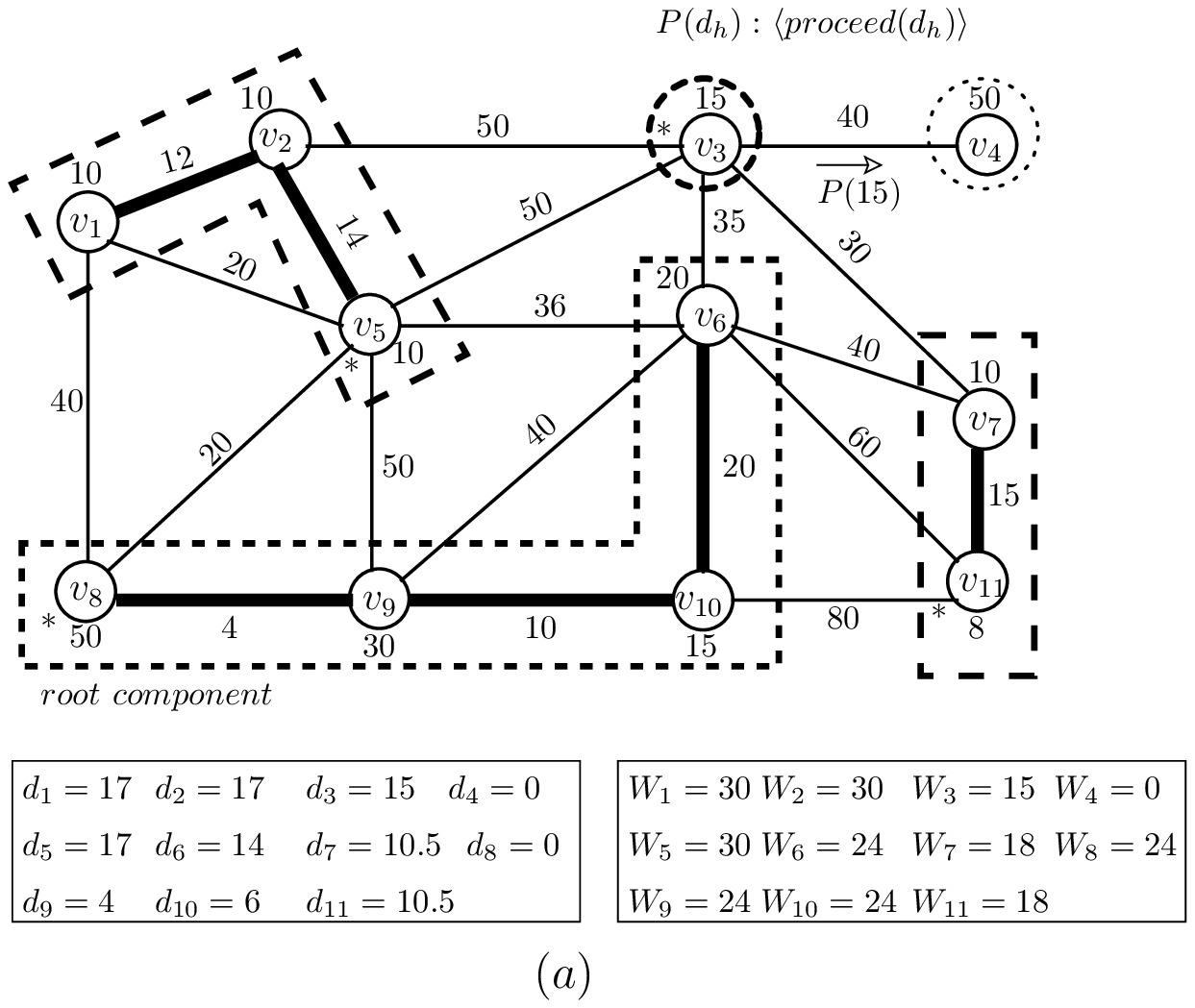}
    \vspace{-10em}
  \end{subfigure}
 \hspace*{-3em}
  \begin{subfigure}[b]{0.5\textwidth}
    \includegraphics[width=10cm]{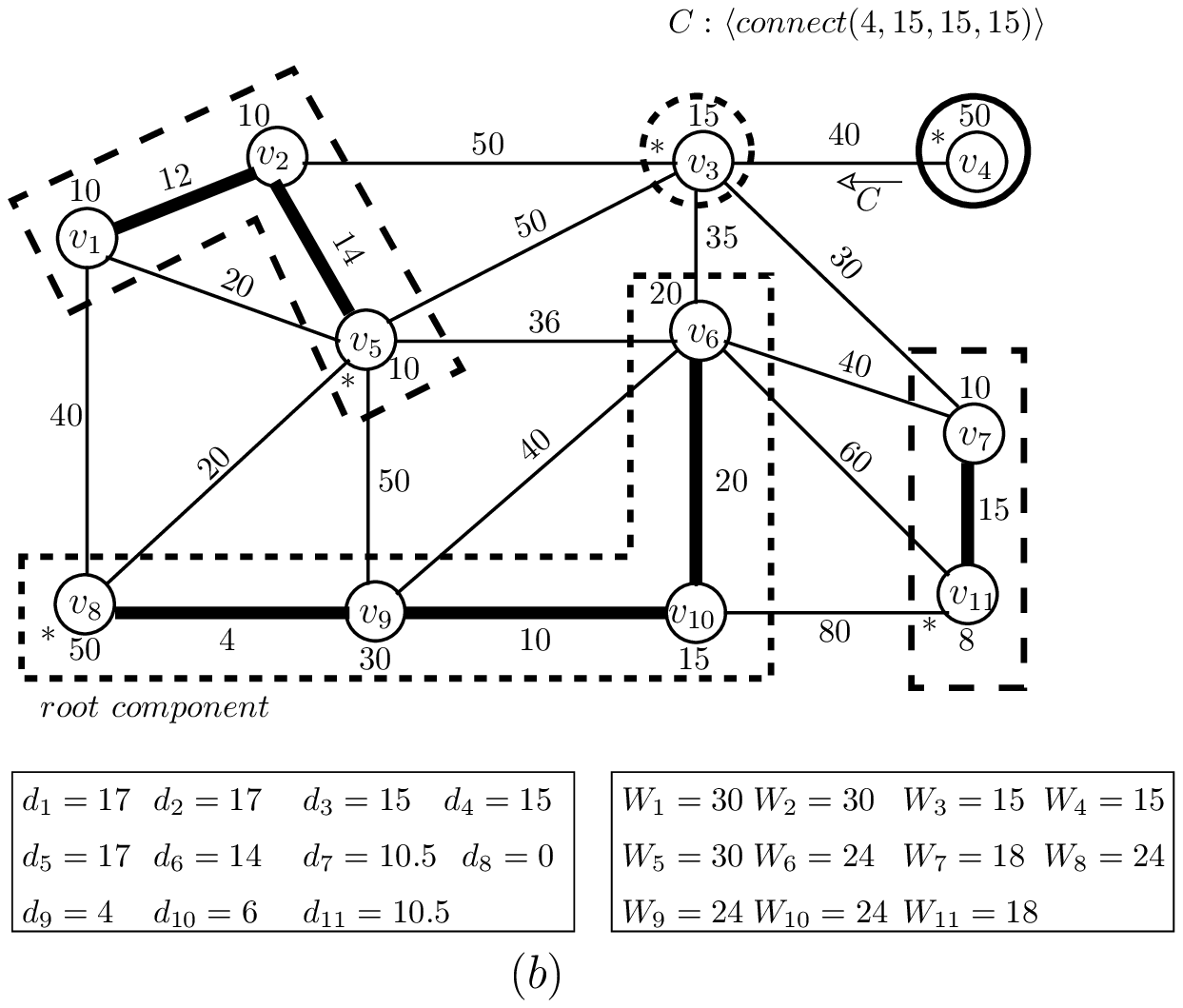}
    \vspace{-10em}
  \end{subfigure}
  \vspace{-2em}
   \captionsetup{font=small}
  \caption{A case of proceed operation. (a) state of sending $\langle proceed(15) \rangle$ by the inactive component $\{v_3\}$. (b) state after the component $\{v_4\}$ receives $\langle proceed(15) \rangle$.} \label{eg-proceed}
\end{figure}

\begin{figure}[!tp]
\vspace{-10.5em}
  \begin{subfigure}[b]{0.5\textwidth}
    \includegraphics[width=10cm]{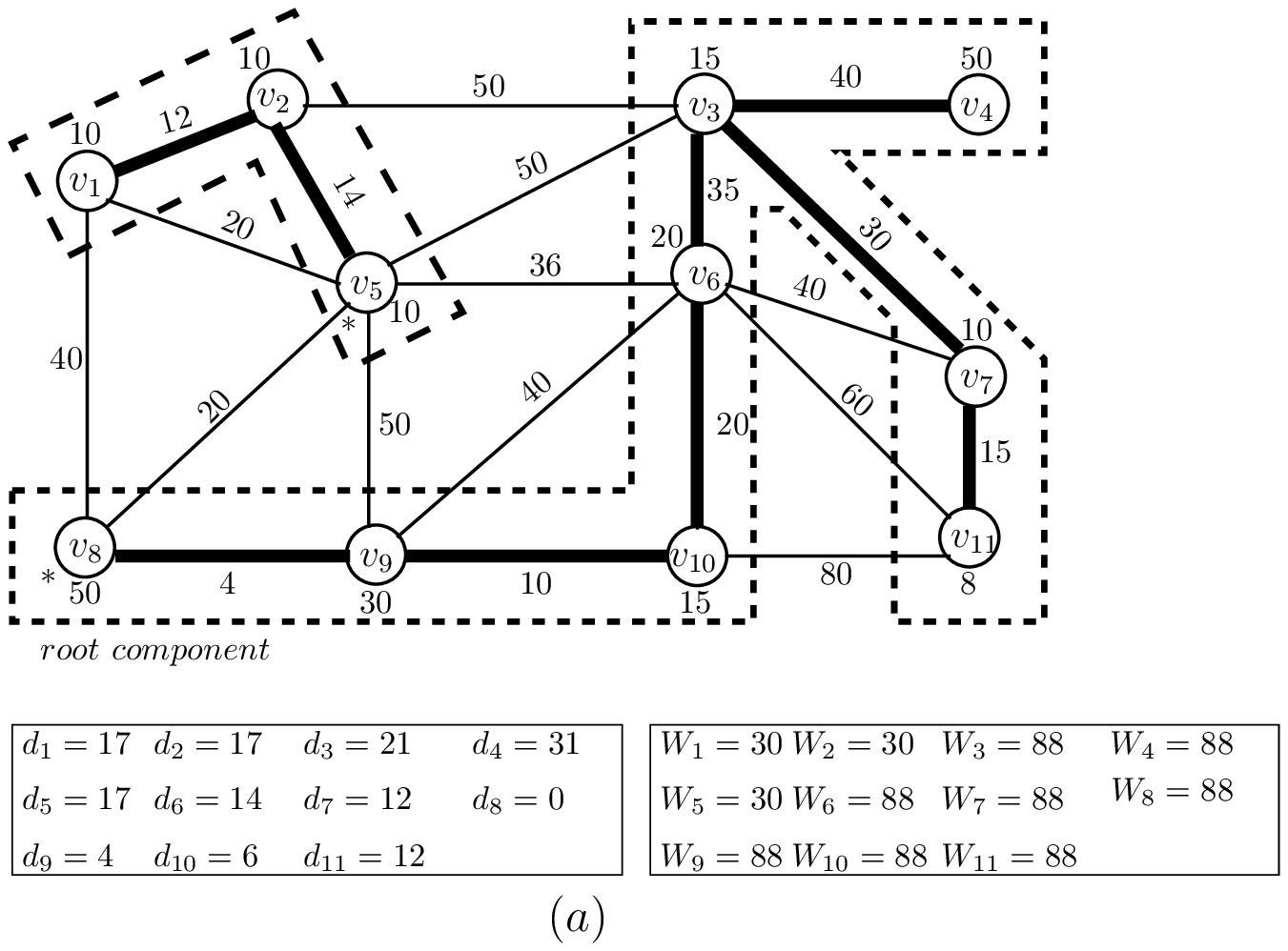}
    \vspace{-10em}
  \end{subfigure}
 \hspace*{-2em}
  \begin{subfigure}[b]{0.5\textwidth}
    \includegraphics[width=10cm]{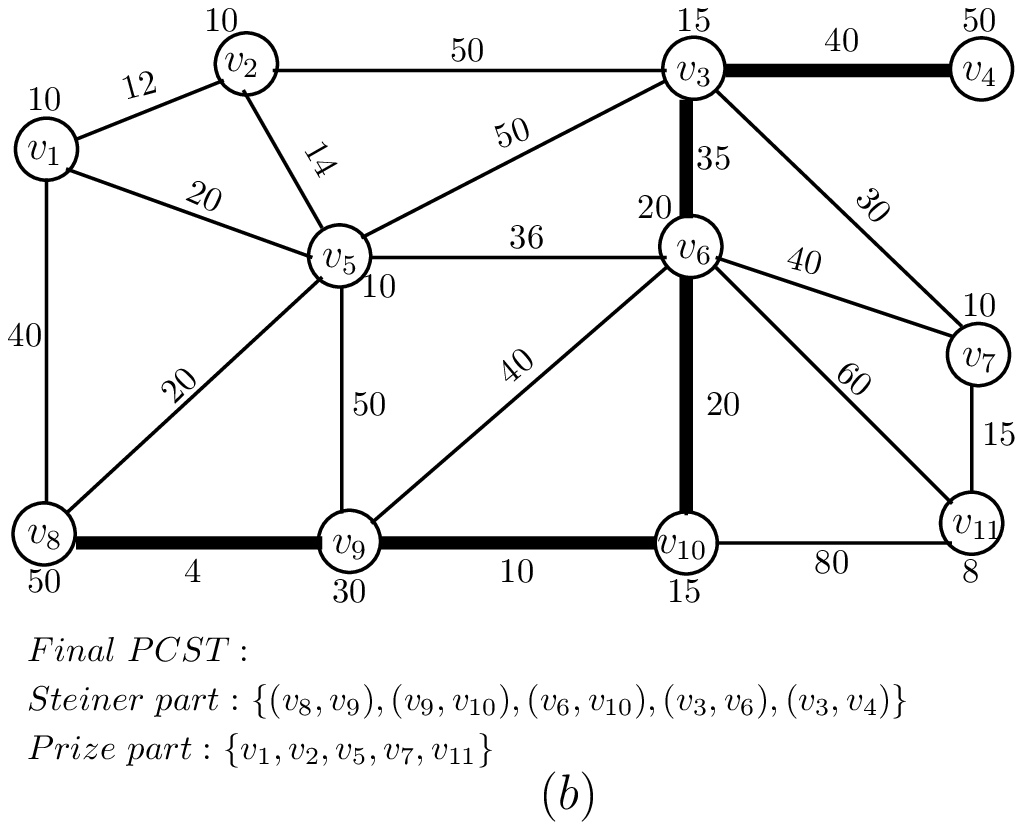}
    \vspace{-10em}
  \end{subfigure}
  \vspace{-3em}
  \captionsetup{font=small}
  \caption{A case of pruning phase. (a) state before pruning. (b) state after pruning phase which is the final solution to the PCST. The pruned components are $\{v_1, v_2, v_5\}$ and $\{v_7, v_{11}\}$.} \label{pruning-phase}
\end{figure}

\section{Proof of Correctness} \label{proof-of-correctness}

\subsection{Termination}
A round of $proc\_initiate()$ in a component $\mathit{C}$ means the time from the beginning of the execution of  $proc\_initiate()$ till the completion of finding  $\epsilon(\mathit{C})$. By an action of an event $A$ we mean the start of event $A$. 
\begin{lemma}
\label{lemma0}
A round of  $proc\_initiate()$ generates at most $6|V| + 2|E| - 4$  messages.
\end{lemma}
\begin{proof}
In our D-PCST algorithm in each round of $proc\_initiate()$, the following messages are possibly generated: $\langle initiate \rangle$, $\langle test \rangle$, $\langle status \rangle$, $\langle reject \rangle$, $\langle report \rangle$, $\langle merge \rangle$, $\langle connect \rangle$,
$\langle update\_info \rangle$, $\langle back \rangle$, $\langle proceed \rangle$, $\langle accept \rangle$, and $\langle refind\_epsilon \rangle$. Since maximum number of $\mathit{branch}$ edges in a component is at most $|V| - 1$, therefore at most $|V| -1$ number of messages are exchanged for each kind of $\langle initiate \rangle$, $\langle report \rangle$, $\langle merge \rangle$, and $\langle update\_info \rangle$ in each round of  $proc\_initiate()$.
Similarly in each round of $proc\_initiate()$, at most $|E|$ number of $\langle test \rangle$ are sent and in response at most $|E|$ number of ($\langle status \rangle$ or $\langle reject \rangle$) messages are generated. The $\langle proceed \rangle$ and the $\langle back \rangle$  are exchanged in between the leaders of two different components. Therefore, for of each kind of $\langle proceed \rangle$ and $\langle back \rangle$ in the worst case at most $|V| - 1$ number of messages are communicated in each round of $proc\_initiate()$. For each $\langle connect \rangle$, either an $\langle accept \rangle$ or a $\langle refind\_epsilon \rangle$ is generated. Therefore in each round of $proc\_initiate()$, at most any two of the combinations of $\{\langle connect \rangle$, $\langle accept \rangle$, $\langle refind\_epsilon \rangle\}$ are generated. Therefore number of messages exchanged in each round of $proc\_initiate()$ is at most $4(|V| - 1)+ 2|E| + 2(|V| - 1) + 2 =6|V| + 2|E| - 4$.
\end{proof}

\begin{lemma}
\label{lemma1}
If the leader of an $\mathit{inactive}$ component $\mathit{C}$ finds $\mathit{\epsilon_1(C) = \infty}$ then for each neighboring component $\mathit{C_k}$ of $\mathit{C}$, $\mathit{CS(C_k) = inactive}$.
\end{lemma} 

\begin{proof}
Suppose by contradiction the leader of $\mathit{C}$ finds $\mathit{\epsilon_1(C) = \infty}$ and there exists a neighboring component $\mathit{C_k}$ of $\mathit{C}$ such that $\mathit{CS(C_k) \neq inactive}$. Therefore either 
$\mathit{CS(C_k) = sleeping}$ or $\mathit{CS(C_k) = active}$. Since it is given that the leader of the component $\mathit{C}$ is in a state of finding its $\mathit{\epsilon_1(C)}$, and $\mathit{CS(C) = inactive}$ therefore Claim~\ref{claim-1} ensures that $\mathit{CS(C_k)\neq active}$ for each neighboring component $\mathit{C_k}$ of $\mathit{C}$.

Consider the case of $\mathit{CS(C_k) = sleeping}$. To find $\mathit{\epsilon_1(C)}$ the leader starts the procedure $proc\_initiate()$ which in turn sends $\langle initiate \rangle$ on each of its $\mathit{branch}$ edges in the component $\mathit{C}$. Upon receiving $\langle initiate \rangle$ each node $\mathit{v \in C}$ forwards it on its outbound $\mathit{branch}$ edges and if $v$ is a  {\em frontier} node  then it also sends $\langle test \rangle$ on each edge $e$ if state of $e$ is neither $\mathit{branch}$ nor $\mathit{rejected}$.    
In response to each $\langle test \rangle$, $v$ either receives $\langle status \mathit{(CS(C_k), d_u)} \rangle$ on an edge $e$ from a node $\mathit{u \in C_k \neq C}$ if $\mathit{C_k}$ is a neighboring component of $\mathit{C}$ or the node $v$ receives $\langle reject \rangle$  if $e = (v, u)$ such that $\mathit{v, u \in C}$. The $\langle reject \rangle$  is simply discarded by the node $v$. In case of $\langle status \mathit{(CS(C_k), d_u)} \rangle$, the node $v$ calculates the $\mathit{\epsilon_e}$ for edge $e$. Since $\mathit{CS(C_k) = sleeping}$ therefore the {\em frontier} node $v$ computes $\mathit{\epsilon_e = w_e - d_v - d_u}$. In this case the computed value $\mathit{\epsilon_e}$ is a finite real number, since $w_e , d_v$ and $d_u$ are finite real numbers. Eventually the value of each $\mathit{\epsilon_1(v)}$ computed by each {\em frontier} node $\mathit{v \in C}$ reaches to the leader of $\mathit{C}$. And the leader of $\mathit{C}$ finds the value of $\mathit{\epsilon_1(C)}$  to be a finite real number, a contradiction to the fact that $\mathit{\epsilon_1(C)=\infty}$. This completes the proof.     
\end{proof}

\begin{claim}
\label{claim-2}
D-PCST algorithm generates the action of $\langle proceed \rangle$ at most $|V| - 1$ times.
\end{claim}
\begin{proof}
During the execution of D-PCST algorithm an $\mathit{inactive}$ component may send $\langle proceed \rangle$ more than once to different components, however a component receives $\langle proceed \rangle$ at most once. In addition the root component is always $\mathit{inactive}$ and never receives $\langle proceed \rangle$. Since there are at most $|V|$ components and root component never receives $\langle proceed \rangle$, therefore at most $|V| - 1$ number of $\langle proceed \rangle$ is generated.
\end{proof}

\begin{lemma}
\label{lemma2}
 D-PCST algorithm generates the action of $\langle back \rangle$ at most $|V| - 1$ times.
\end{lemma}
\begin{proof}
	A non-root component $\mathit{C}$ decides to take the action of $\langle back \rangle$ only if $\mathit{CS(C) = inactive}$ and it finds its $\mathit{\epsilon(C) = \infty}$.  Lemma~\ref{lemma1} proves that if $\mathit{CS(C) = inactive}$ and $\mathit{\epsilon_1(C) = \infty}$ then $\mathit{CS(C_k) = inactive}$ for each neighboring component $\mathit{C_k}$ of $\mathit{C}$. In addition, $\mathit{SE(e) \neq refind}$ for each edge $\mathit{e \in \delta(C)}$. These facts indicate that all neighboring components of $\mathit{C}$ are explored and therefore action of $\langle back \rangle$ is taken by the leader node in search of a component whose state is still $\mathit{sleeping}$. The component $\mathit{C}$ sends $\langle back \rangle$ to a neighboring component say $\mathit{C'}$ which sent $\langle proceed \rangle$ to $\mathit{C}$ or to a subcomponent of $C$ in some early stages of the algorithm. Since each $\langle proceed \rangle$ generates the action of $\langle back \rangle$ at most once and by Claim~\ref{claim-2} D-PCST algorithm generates the action of $\langle proceed \rangle$ at most $|V| - 1$, therefore D-PCST algorithm generates the action of $\langle back \rangle$ at most $|V| - 1$ times. 
\end{proof}

\begin{claim}
\label{claim-1}
 When an inactive component $\mathit{C}$ is in the state of computing its $\mathit{\epsilon_1(C)}$ then there cannot exist any component $\mathit{C'}$ in the neighborhood of $\mathit{C}$ such that $\mathit{CS(C')=active}$.
\end{claim}


\begin{claim}
\label{claim-3}
If none of the four consecutive rounds of $proc\_initiate()$ initiates the action of sending $\langle back \rangle$ then any one of the following two events is guaranteed to happen: (i) sum of the number of components decreases (ii) one of the sleeping or active components decreases.
\end{claim}
\begin{proof}
In our proposed D-PCST algorithm the leader of a component $\mathit{C}$ starts finding its $\mathit{\epsilon(C)}$ by executing the procedure $proc\_initiate()$. Depending on the current state of $\mathit{C}$ i.e. $\mathit{CS(C)}$ and its computed value of $\mathit{\epsilon(C)}$, the leader of $\mathit{C}$ decides to take any one of the following actions: (i) {\em merging} (ii) {\em deactivation} (iii) {\em  sending $\langle proceed \rangle$}, (iv) {\em sending $\langle back \rangle$}, and (v) {\em pruning}. If the action of sending {\em $\langle back \rangle$} is not taken by any one of the four consecutive rounds of $proc\_initiate()$ then within four consecutive rounds of the $proc\_initiate()$ any one of the remaining actions is guaranteed to happen. 

First consider the action of {\em merging}. Before the action of {\em merging}, the leader of the component $\mathit{C}$ computes its $\mathit{\epsilon(C)}$ in one round of $proc\_initiate()$. After that it sends a $\langle merge \rangle$ to the corresponding component say $\mathit{C'}$. If $\mathit{CS(C') = inactive}$ then $\mathit{C'}$ immediately merges with $\mathit{C}$ and in this case merge happens in one round of $proc\_initiate()$. As a result one of the components decreases in the graph. If $\mathit{CS(C') = sleeping}$ then $\mathit{C'}$ takes one round of $proc\_initiate()$ to decide whether to merge with $\mathit{C}$ or deactivate itself. If it decides to merge with $\mathit{C}$ then number of component decreases by one in the graph. On the other hand if it decides to deactivate itself then a $\mathit{sleeping}$ component vanish in the graph. Therefore action of {\em merging} takes at most two rounds of $proc\_initiate()$.

We know that only an $\mathit{active}$ component can decide to deactivate itself. To be deactivated a component $\mathit{C}$ finds its $\mathit{\epsilon(C) = \epsilon_2(C)}$ in exactly one rounds of  $proc\_initiate()$. As a result one $\mathit{active}$ component decreases in the graph.

A component $\mathit{C}$ initiates the action of  sending $\langle proceed \rangle$ only if it is in $\mathit{inactive}$ state. For this action, first the leader of $C$ computes its $\mathit{\epsilon_1(C)}$ which takes one round of $proc\_initiate()$. After that it sends $\langle proceed \rangle$ to the corresponding neighboring component say $\mathit{C'}$. Upon receiving $\langle proceed \rangle$ from $\mathit{C}$, depending on its current state, the component $\mathit{C'}$ does the following: 
\begin{enumerate}[(i)]
\item  $\mathit{CS(C') = sleeping}$. $\mathit{C'}$ starts finding its $\mathit{\epsilon(C')}$ to decide whether to {\em merge} with some other component or (ii) {\em deactivate} itself. In case of merging, it takes at most another two rounds of $proc\_initiate()$ and as a result one component decreases, i.e. from the point of finding $\mathit{\epsilon_1(C)}$ at $\mathit{C}$ upto the merging of the component $\mathit{C'}$ with a neighboring component it takes at most three rounds of $proc\_initiate()$. In case of deactivation,  $\mathit{C'}$ takes exactly one round of $proc\_initiate()$ for which one sleeping component decreases and takes total two rounds of $proc\_initiate()$ from the point of finding $\mathit{\epsilon_1(C)}$ at $\mathit{C}$ upto the deactivation of $\mathit{C'}$.

\item $\mathit{CS(C') = inactive}$. In this case $\mathit{C'}$ receives $\langle proceed \rangle$ because there exists an edge $e$ such that $e \in \delta(\mathit{C'}) \wedge \delta(\mathit{C})$ and the state the edge $e \in \delta(v)$ at some node $\mathit{v \in C}$ must be $\mathit{refind}$. This is because in some early stages of the algorithm the component $\mathit{C}$ or a sub-component of $\mathit{C}$ sent a $\langle connect \rangle$ to $\mathit{C'}$ and in response to that, $\mathit{C'}$ became $\mathit{inactive}$ and as a result sent back a $\langle \mathit{refind\_epsilon} \rangle$ to $\mathit{C}$ or to the sub-component of $\mathit{C}$. Now  there should be at least one component $\mathit{C''}$ such that $\mathit{CS(C'') = sleeping}$ in the neighborhood of $\mathit{C'}$. Otherwise $C'$ has to take the action of {\em sending} $\langle back \rangle$ which is not possible according to our assumption. Since $\mathit{CS(C') = inactive}$ therefore $\mathit{C'}$ takes one round of $proc\_initiate()$ to compute its $\mathit{\epsilon_1(C')}$ to take the action of {\em sending $\langle proceed \rangle$} to a neighboring component $\mathit{C''}$ such that $\mathit{CS(C'') = sleeping}$. After that $\mathit{C''}$ follows at most two rounds of $proc\_initiate()$ to decide the action of either {\em merging} or  {\em deactivation} which guarantees the occurring of any one of the mentioned events. Therefore  from the point of finding $\mathit{\epsilon_1(C)}$ at $\mathit{C}$ upto any one of the events to be happened takes at most four rounds of $proc\_initiate()$.
\item $\mathit{CS(C') = active}$. By Claim~\ref{claim-1} this condition is not possible.
\end{enumerate}

In case of the action of {\em pruning}, the root component $\mathit{C_r}$ takes exactly one round of $proc\_initiate()$ to computes its $\mathit{\epsilon(C_r)}$ which must be equal to $\infty$.
Therefore, we claim that if none of the four consecutive rounds of $proc\_initiate()$ initiates the action of {\em sending $\langle back \rangle$} then any one of the following events is guaranteed to happen: (i) sum of the number of components decreases (ii) one of the $\mathit{sleeping}$ or $\mathit{active}$ components decreases.
\end{proof}

\begin{lemma}
\label{lemma3}
The {\em growth} phase of the D-PCST algorithm terminates after at most $9|V| -7$ rounds of   $proc\_initiate()$.  
\end{lemma}
\begin{proof}
Initially the state of the root component $\mathit{C_r}$ is $\mathit{inactive}$ and it takes one round of  $proc\_initiate()$  to compute its $\mathit{\epsilon_1(C_r)}$. After that $\mathit{C_r}$ sends $\langle proceed \rangle$ to a neighboring component to take further actions of the algorithm. Claim~\ref{claim-3}  ensures that in the worst case at most $4(|V| - 1)$ round of $proc\_initiate()$ is required to decrease the sum of the number of components and becomes one or at most $4(|V| - 1)$ rounds of $proc\_initiate()$ is required to change the state of each $\mathit{sleeping}$ or $\mathit{active}$ component to $\mathit{inactive}$ state. Lemma~\ref{lemma2} proves that the action of $\langle back \rangle$ is generated at most $|V| - 1$ times. If the root node $r \in C_r$ finds that $\mathit{\epsilon_1(C_r) = \infty}$ then instead of taking the the action of $\langle back \rangle$ the root component starts the {\em pruning} phase which indicates the termination of the {\em growth} phase. Before the termination of the {\em growth} phase additionally one round of $proc\_initiate()$ is required to find $\mathit{\epsilon_1(C_r) = \infty}$. Summing for all the cases we get that the total number of rounds of $proc\_initiate()$ is equal to $1 + 4(|V| - 1) + 4(|V| - 1) + (|V| - 1) + 1 = 9|V| - 7$. Therefore it is guaranteed that after at most $ 9|V| - 7$ rounds of $proc\_initiate()$ the initiation of round of $proc\_initiate()$ stops. Once the  $proc\_initiate()$ stops, no more messages related to the {\em growth} phase are exchanged in the network. This ensures that the {\em growth} phase eventually terminates after at most $9|V| -7$ rounds of $proc\_initiate()$. 
\end{proof}

\begin{lemma}
\label{lemma4}
Pruning phase of D-PCST algorithm eventually terminates.
\end{lemma}

\begin{proof}
After the termination of the {\em growth} phase the root node $r$ initiates the {\em pruning} phase by sending $\langle prune \rangle$ on each edge $\mathit{e \in \delta(r)}$ if $\mathit{SE(e) = branch}$ or  $\mathit{EPM(e) = TRUE}$. Here at node $\mathit{v \in C}$ in a component $C$, $\mathit{EPM(e)}$ is a local boolean variable for each edge $\mathit{e \in \delta(v)}$ and by default $\mathit{EPM(e) = FALSE}$ for each $\mathit{e \in \delta(v)}$. Whenever a {\em frontier} node sends a $\langle proceed \rangle$ over an incident edge $e$ to some other neighboring component then it sets $\mathit{EPM(e) = TRUE}$.  Upon receiving $\langle prune \rangle$ on an edge $e$, except on edge $e$ a node $v$ forwards $\langle prune \rangle$ on each edge $e' \in \delta(v)$ if $\mathit{SE(e') = branch}$ or  $\mathit{EPM(e') = TRUE}$. After that if $v$ belongs to a non-root inactive component then it sets $\mathit{SE(e'') = basic}$ for each edge $e'' \in \delta(v)$ if $\mathit{SE(e'') \neq basic}$. Note that a node receives $\langle prune \rangle$  exactly once.

In the root component $\mathit{C_r}$, pruning starts at leaf nodes of the tree rooted at $r$. Whenever a leaf node $\mathit{v \in C_r}$ receives $\langle prune \rangle$  then $v$ prunes itself if $\mathit{labelled\_flag = TRUE}$ and there exists exactly one edge $\mathit{e \in \delta(v)}$ such that $\mathit{SE(e) = branch}$. In this case, $v$ sets its local variable $\mathit{prize\_flag}$ to $\mathit{TRUE}$ indicating that it is contributing to the Penalty part of the PCST and $\mathit{labelled\_flag = FALSE}$. Then $\langle backward\_prune \rangle$ is sent on $e$ and $\mathit{SE(e)}$ is set to $\mathit{basic}$. Upon receiving $\langle backward\_prune \rangle$ on an edge say $e$, a node first sets $\mathit{SE(e)=basic}$ and continues with the pruning operation. If a node $v$ fails to prune itself then no further message is sent on any of its incident edges. In this way all nodes in the $C_r$ stops sending further messages on their incident edges and thus  {\em pruning} phase eventually terminates.  
\end{proof}

\begin{theorem}
\label{theorem1}
The D-PCST algorithm eventually terminates.
\end{theorem}

\begin{proof}
Lemma~\ref{lemma3} and Lemma~\ref{lemma4} prove that the {\em growth} phase and the {\em pruning} phase of the D-PCST algorithm terminate respectively. Together Lemma~\ref{lemma3} and Lemma~\ref{lemma4} prove that the D-PCST algorithm terminates. 
\end{proof}

\subsection{Message and Time complexity}
\subsubsection{A lower bound on message complexity for trivial distributed PCST}
\label{subsec:message_complexity}
A trivial message passing distributed algorithm for PCST can be as follows. It  collects the whole information of the network (weight of each $e \in E$, ID of each $v \in V$, and $p_v$ for each $v \in V$) at some node ($r$), computes the PCST using the  $GW$-algorithm \cite{GW_1995} at $r$ and then  $r$ informs each  node whether it belongs to the Penalty part or the Steiner part of the PCST. 
It is obvious that this is a non-local algorithm and suffers from single point of failure. However in the absence of a benchmark local algorithm for distributed PCST, it serves the purpose of  calculating a trivial lower bound of message complexity. 
For message-efficient convergecast and broadcast of information,  a spanning tree (similar to BFS or DFS) of the network can be constructed rooted at $r$. The lower bound of message complexity for finding a rooted spanning tree using distributed algorithm is $\Omega(|E|)$ \cite{Awerbuch:1987:ODA:28395.28421}.  Now $r$ broadcasts a query message asking each node to send their local information. This broadcast requires $O(|V|)$ messages. Each node sends its information to $r$ using a path. The path  contains $O(|V|)$ intermediate nodes. Therefore collecting information from all the nodes to $r$ requires $O(|V|^2)$ messages. Note here that we do not use any message aggregation at 
intermediate nodes  in this trivial algorithm.  Upon receiving the information from all the nodes, $r$ can compute the PCST of the entire graph using the centralized $GW$-algorithm.  After the computation of the PCST, $r$  sends message to each node  informing whether it belongs to the Penalty part or the Steiner part. In case a node belongs to the Steiner, $r$ also informs the node about the specific incident edges which are part of the Steiner. Like the earlier convergecast, this broadcast of information will take $O(|V|^2)$ messages. Therefore a trivial distributed PCST algorithm as described above takes $\Omega(|E|  + |V|^2)$   messages in a general networks. In a connected graph, $|V| - 1 \leq |E| \leq |V|^2$. Therefore a trivial lower bound of message complexity for distributed PCST is $\Omega(|V|^2)$. However the time complexity remains $O(|V|^2 \log |V|)$ since at the central node we use GW-algorithm. 

\subsubsection{Message Complexity of the D-PCST algorithm}
We determine here the upper bound on the number of messages exchanged during the execution of the D-PCST algorithm. Lemma~\ref{lemma3} shows that the {\em growth} phase terminates after at most $9|V| - 7$ rounds of $proc\_initiate()$. And Lemma~\ref{lemma0} proves that in the worst case a round of $proc\_initiate()$ can generate at most $6|V| + 2|E| - 4$  messages. Therefore number of messages exchanged until the termination of the {\em growth} phase is at most $(9|V| -7)(6|V| + 2|E| - 4)$. 
	
In the {\em pruning} phase two types of messages are generated namely $\langle \mathit{prune} \rangle$ and $\langle \mathit{backward\_prune} \rangle$. A node $v$ sends or forwards $\langle prune \rangle$ on an edge $e \in \delta(v)$ if $\mathit{SE(e) = branch}$ or $\mathit{EPM(e) = TRUE}$. We know that the number of $\mathit{branch}$ edges of a component $\mathit{C}$ is exactly $|\mathit{C}| - 1$. Whenever the {\em growth} phase terminates then the total number of $\mathit{branch}$ edges of all the components in the graph is at most $|V| - 1$. It follows that at most $|V| - 1$ $\langle \mathit{prune} \rangle$ are generated with respect to $\mathit{branch}$ edges. Similarly at node $v$, the local boolean variable $\mathit{EPM(e)}$ is set to $\mathit{TRUE}$ if it sends a $\langle \mathit{procced} \rangle$ on edge $e$. Since by Claim~\ref{claim-2} at most $|V| - 1$ $\langle \mathit{procced} \rangle$ are sent in D-PCST and for each edge $e$ on which $\langle \mathit{procced} \rangle$ is sent the variable $\mathit{EPM(e)}$ is set to $\mathit{TRUE}$, therefore at most $|V| - 1$ $\langle \mathit{prune} \rangle$ are sent with respect to the boolean variable $\mathit{EPM}$. Therefore at most $2|V| - 2$ number of  $\langle \mathit{prune} \rangle$ is generated in the {\em pruning} phase.
A $\langle \mathit{backward\_prune} \rangle$ is exchanged only within the root component. Upon receiving $\langle prune \rangle$ a leaf node (which has exactly one incident $\mathit{branch}$ edge) of the root component sends a $\langle backward\_prune \rangle$ on the $\mathit{branch}$ edge only if it decides to prune itself. Since possible number of $\mathit{branch}$ edges in the root component is at most $|V| - 1$, therefore at most $|V| - 1$ number of $\langle \mathit{backward\_prune} \rangle$ is generated in the {\em pruning} phase. It follows that at most $3(|V| - 1)$ messages are generated in the {\em pruning} phase.

The analysis shown above proves that number of messages exchanged until the termination of the D-PCST algorithm is at most $(9|V| -7)(6|V| + 2|E| - 4) + 3(|V| - 1)$ which is equivalent to $O(|V|^2 + (|E||V|)$. Since the graph is connected, therefore  $|V| - 1 \leq |E| \leq |V|^2$. This implies that $|V|^2 = O(|E||V|)$. Using these facts we get that message complexity of the D-PCST algorithm is $O(|E||V|)$. Therefore we can claim the following theorem.
\begin{theorem}
\label{theorem2}
The message complexity of the D-PCST algorithm is $O(|E||V|)$.
\end{theorem}
{\bf Time Complexity}. The worst case time complexity of the proposed algorithm is $O(|V||E|)$, which can be fine-tuned to give a complexity of $O(\mathcal{D}|E|)$ where $\mathcal{D}$ is the diameter of the network.

\subsection{Optimality of the D-PCST algorithm}
\begin{lemma}
\label{out-lemma1}
If $\mathit{CS(C_l) = sleeping}$ and it receives $\langle \mathit{connect(v, W(C_k), d_v, d_h(C_k))} \rangle$ or $\langle \mathit{proceed(d_h(C_k))} \rangle$  over an edge $e$ from a node $\mathit{v\in C_k}$ where $\mathit{C_k}$ is a neighboring component of $C_l$ then $C_l$  correctly computes  each of its local variables without violating any of the dual constraints.
\end{lemma}

\begin{proof}
	Since $\mathit{CS(C_l) = sleeping}$, therefore $\mathit{C_l}$ is a single node component. Let it be $\{u\}$. If $u$ receives $\langle \mathit{connect(v, W(C_k), d_v, d_h(C_k))} \rangle$ from a node $\mathit{v \in C_k}$  over the edge $e$ then first  $\mathit{C_l}$ becomes {\em active} and then  $u$ initializes its local variables $\mathit{d_u = d_h(C_k)}$, $\mathit{W(C_l) = d_h(C_k)}$ and $\mathit{d_h(C_l) = d_h(C_k)}$. After that  $u$ computes $\mathit{\epsilon(C_l)}$ as follows:  
	\begin{align*}
		\mathit{\epsilon_e} = & \mathit{\frac{w_e - d_u - d_v}{2}}\\
		\mathit{\epsilon_2(C_l)} = & \mathit{TP(C_l) - W(C_l)} = \mathit{p_u - W(C_l)}
	\end{align*}
	and $\mathit{\epsilon(C_l) = \min (\epsilon_e, \epsilon_2(C_l))}$. Therefore it is clear that $\mathit{d_h(C_k)}$ is used by $\mathit{C_l}$ to compute its $\mathit{\epsilon(C_l)}$ in case $\mathit{CS(C_l)=sleeping}$. Now there are four possible cases:\\
\textbf{Case 1}: {$\mathit{\epsilon(C_l) = \epsilon_e}$ and $\mathit{\epsilon(C_l)}<0$.} {The condition
   $\mathit{\epsilon(C_l)}<0$ indicates that the dual {\em edge packing} constraint $\mathit{\sum_{S: e \in \delta(S)} y_S \leq w_e}$ is violated on the edge $e$ when each dual variable $\mathit{y_S : S \subset V \wedge e \in \delta(S)}$ is increased by a value $\mathit{\epsilon(C_l)}$. More specifically dual variables $\mathit{y_{C_l}}$ and $\mathit{y_{C_k}}$ are excessively increased by $\mathit{\epsilon(C_l)}$. To ensure that the dual constraint is not violated, the excess value $\mathit{\epsilon(C_l)}$ must be deducted from each of the dual variables $\mathit{y_{C_l}}$ and $\mathit{y_{C_k}}$. After the deduction, both components $\mathit{C_l}$ and $\mathit{C_k}$  merge and form a new component $\mathit{C_l \cup C_k}$ without violating the dual constraints. Note that every node $\mathit{v \in C_l \cup C_k}$ also updates its local variables $\mathit{d_v = d_v - \epsilon(C_l)}$ and $\mathit{W(C_l \cup C_k) = W(C_l) + W(C_k) - 2\epsilon(C_l)}$. In addition, corresponding $\mathit{d_h(C_l \cup C_k)}$ is also updated accordingly.
    }\\
\textbf{Case 2}: {$\mathit{\epsilon(C_l) = \epsilon_e}$ and $\mathit{\epsilon(C_l)}\geq 0$.}{ This ensures that at most $\mathit{\epsilon(C_l)}$ can be added to both $\mathit{y_{C_l}}$ and  $\mathit{y_{C_k}}$ without violating the dual {\em edge packing} constraint $\mathit{\sum_{S: e \in \delta(S)} y_S \leq w_e}$ for edge $e$. Therefore the components $\mathit{C_l}$ and $\mathit{C_k}$ merge through the edge $e$ and forms a bigger component $\mathit{C_l \cup C_k}$. Each node $\mathit{v \in C_l \cup C_k}$ also updates  its local variables $\mathit{d_v = d_v + \epsilon(C_l)}$ and $\mathit{W(C_l \cup C_k) = W(C_l) + W(C_k) + 2\epsilon(C_l)}$. In addition, corresponding $\mathit{d_h(C_l \cup C_k)}$ is also updated accordingly.  
    }\\
\textbf{Case 3}: {$\mathit{\epsilon(C_l) = \epsilon_2(C_l)}$ and $\mathit{\epsilon(C_l)}<0$. }{In this case the dual variable $\mathit{y_{C_l}}$ for the component $\mathit{C_l}$ is excessively increased by $\mathit{\epsilon(C_l)}$ and this indicates that the dual {\em penalty packing} constraint $\mathit{\sum_{S \subseteq C_l} y_S \leq \sum_{v \in C_l} p_v}$ is violated at $\mathit{C_l}$. Therefore $\mathit{y_{C_l} = y_{C_l} - \epsilon_2(C_l)}$ and it ensures that the dual {\em penalty packing} constraint for the component $\mathit{C_l}$ is not violated and becomes tight. Node $\mathit{u \in C_l}$ updates its local variables $\mathit{d_u = d_u - \epsilon(C_l)}$ and $\mathit{W(C_l) = W(C_l) - \epsilon(C_l)}$. In addition, corresponding $\mathit{d_h(C_l)}$ is also updated accordingly.
}\\
\textbf{Case 4} : {$\mathit{\epsilon(C_l) = \epsilon_2(C_l)}$ and $\mathit{\epsilon(C_l)}\geq 0$.} {This indicates that at most $\mathit{\epsilon(C_l)}$ can be added to the dual variable $\mathit{y_{C_l}}$ in component $\mathit{C_l}$  without violating the dual {\em penalty packing} constraint $\mathit{\sum_{S \subseteq C_l} y_S \leq \sum_{v \in C_l} p_v}$. Since after the addition of $\mathit{\epsilon(C_l)}$ to the dual variable $\mathit{y_{C_l}}$, the dual {\em penalty packing} constraint $\mathit{\sum_{S \subseteq C_l} y_S \leq \sum_{v \in C_l} p_v}$ becomes tight, therefore the component $\mathit{C_l}$ decides to deactivate itself. The node $\mathit{u \in C_l}$ updates its local variables $\mathit{d_u = d_u + \epsilon(C_l)}$ and $\mathit{W(C_l) = W(C_l) + \epsilon(C_l)}$. In addition, corresponding $\mathit{d_h(C_l)}$ is also updated accordingly.} Therefore after receiving $\langle connect \rangle$, $\mathit{C_l}$ correctly computes each of its local variable without violating any of the dual constraints.
  
  Similarly if $u$ receives $\langle proceed\mathit{(d_h(C_k))} \rangle$ from a node $\mathit{v \in C_k}$ over an edge $e$ then first $\mathit{C_l}$ becomes {\em active} and then node $u$ initializes its local variables $\mathit{d_u = d_h(C_k)}$, $\mathit{W(C_l) = d_h(C_k)}$ and $\mathit{d_h(C_l) = d_h(C_k)}$. Note that if a component receives a $\langle proceed \rangle$ then the state of each component in its neighborhood is either $\mathit{sleeping}$ or $\mathit{inactive}$. In this state of the algorithm there can not exist any active component in the whole network. 
Let $\mathit{e' \in \delta(C_l)}$ be the MOE of $\mathit{C_l}$ which connects to a node $\mathit{w \in C_p \neq C_l}$. Then $\mathit{C_l}$ computes its $\mathit{\epsilon(C_l)}$ as follows:
	\begin{align*}
		{\epsilon_1(C_l)} = & 
			\begin{cases} 
						\mathit{\frac{w_{e'} - d_u - d_h(C_l)}{2}}, & \text{if}\  \mathit{CS(C_p) = sleeping}\\
						\mathit{w_{e'} - d_u - d_w}			, & \text{if}\ \mathit{CS(C_p) = inactive}\\
      		\end{cases}\\
		\mathit{\epsilon_2(C_l)} \ = & \ \mathit{TP(C_l) - W(C_l)} = \mathit{p_u - W(C_l)}
	\end{align*}
and $\mathit{\epsilon(C_l) = \min (\epsilon_1(C_l), \epsilon_2(C_l))}$. Therefore it is clear that if $\mathit{CS(C_l) = sleeping}$ then $\mathit{C_l}$ uses $\mathit{d_h(C_k)}$ to compute its $\mathit{\epsilon(C_l)}$. Now following the same way as we have shown for the case of receiving  $\langle connect\rangle$, it can be shown that upon receiving $\langle proceed\mathit{(d_h(C_k))} \rangle$,  $\mathit{C_l}$ correctly computes each of its local variable without violating any of the dual constraints.
\end{proof}

We claim that the approximation factor achieved by our distributed algorithm on a graph of $n$ nodes is $(2 - \frac{1}{n-1})$ of the optimal (OPT). This can be proved from the facts that $d_v = \sum_{S: v \in S} y_S$ for each node $v \in V$ and $W(C) = \sum_{S \subseteq C} y_S$ for each component $C$. Let $OPT_{LP}$ and $OPT_{IP}$ be the optimal solutions to (LP) and (IP) of PCST problem respectively. Then it is obvious that $\sum_{S \subset V} y_S \leq OPT_{LP} \leq OPT_{IP}$.

\begin{theorem}[Goemans and Williamson, \cite{GW_1995}]
\label{theorem:gw-pcst1}
	D-PCST algorithm selects a set of edges $F'$ and a set of vertices $X$  such that 
	\begin{equation} \label{DPCST-eq1}
		\sum_{e \in F'} w_e + \sum_{v \in X} p_v \leq (2 - \frac{1}{n-1}) \sum_{S \subset V} y_S \leq (2 - \frac{1}{n-1}) OPT_{IP} 
		\end{equation} 
		 where $n = |V|$ and $OPT_{IP}$ is the optimal solution to the IP of the PCST. 
\end{theorem}

The above theorem is in fact a transcript of the analogous theorem of Goemans and Williamson \cite{GW_1995} to the distributed setting. We present below the analysis of the approximation ratio for the distributed setting. The main challenge in this proof is to preserve the dual packing constraints in a distributed way so that approximation ratio is preserved.

\begin{proof}
In the construction of $\mathit{F'}$ if a node $v \in V$ is not covered by $\mathit{F'}$ then $v$ must belong to some component deactivated at some point of execution of the algorithm. Let $\mathit{X = \{C_1, C_2,....C_z\}}$ is the set of deactivated components whose nodes are not covered by $\mathit{F'}$. Therefore $\mathit{X}$ can be considered as a set of disjoin subsets of vertices and each subset is some $\mathit{C_j}$ for $j: 1 \leq j \leq z$. Since each $\mathit{C_j}$ is a deactivated component, therefore it follows the fact that $\mathit{\sum_{S \subseteq C_j} y_S = \sum_{v \in C_j} p_v}$. For each edge $\mathit{e \in F'}$ it also follows that $\mathit{\sum_{S: e \in \delta(S)} y_S = w_e}$ and this implies $\mathit{\sum_{e \in F'} w_e = \sum_{e \in F'} \sum_{S: e \in \delta(S)} y_S}$. Putting these in the inequality ~(\ref{DPCST-eq1}) we get
 	\begin{equation} \label{DPCST-eq2}
		\mathit{\sum_{e \in F'} \sum_{S: e \in \delta(S)} y_S + \sum_{j} \sum_{S \subseteq C_j} y_S} \leq (2 - \frac{1}{n-1}) \mathit{\sum_{S \subset V} y_S}
		\end{equation} 
  	  	
  	 Now it can be shown by the method of induction that for each $\mathit{\epsilon(C)}$ computed by a component $\mathit{C}$, the inequality (\ref{DPCST-eq2}) always holds. Here we show for the case of $\mathit{\epsilon(C)} > 0$. 
  	
  	At the beginning of the algorithm the inequality (\ref{DPCST-eq2}) holds since $\mathit{F' = \phi}$, the component containing the root node $r$ is the only trivial single node tree and $\mathit{y_C} = 0$ for each single node component $\mathit{C}$. Let $\mathbb{C}$ is the set of components in the graph when a component $\mathit{C}$ computes its $\mathit{\epsilon(C)}$. Components of $\mathbb{C}$ are categorized into two types of components namely type $A$ and type $I$ as follows:
\begin{itemize}  	
\item A component $\mathit{C'} \in \mathbb{C}$ is denoted as type $A$ if  $\mathit{CS(C') = active}$ {\bf or} $\mathit{CS(C')} = \mathit{sleeping} \wedge \mathit{TP(C') > d_h(C)}$.
\item A component $\mathit{C'} \in \mathbb{C}$  is denoted as type $I$ if $\mathit{CS(C') = inactive}$ {\bf or}  $\mathit{CS(C')} = \mathit{sleeping} \wedge \mathit{TP(C') \leq d_h(C)}$.
\end{itemize}  
The type of a component $\mathit{C}$ is denoted by $\mathit{type(C)}$.
To show that the inequality (\ref{DPCST-eq2}) always holds first we construct a special graph termed as $H=(V', E')$. The set of components of $\mathbb{C}$ is considered as the set of vertices $V'$ of the graph $H$. The vertex set $V'$ contains two types of vertices namely type $A$ and type $I$. The set of edges is $E'= \{\mathit{ e \in (\delta(C') \cap F')} : \mathit{type(C')= A} \}$. All isolated vertices of type $I$ are discarded from the graph $H$. Let $\mathit{N_A}$ denotes the set of vertices of type $A$, $\mathit{N_I}$ denotes the set of vertices of type $I$, $\mathit{N_D}$ denotes the set of vertices of type $A$ such that each vertex of  $\mathit{N_D}$ corresponds to some $\mathit{C_j}$ for $j: 1 \leq j \leq z$, and  $d_v$ denotes the degree of a vertex $v$ in graph $H$. Note that degree of each vertex $\mathit{v \in N_D}$ is zero, i.e. $\mathit{N_D = \{v \in N_A : d_v} = 0\}$. For each $\mathit{\epsilon(C)} > 0$, maximum increment in the left hand side of the inequality (\ref{DPCST-eq2}) is $\mathit{\sum_{v \in N_A} \epsilon_vd_v + \sum_{v \in N_D}\epsilon_v}$, where $\mathit{\epsilon_v} \in (0, \epsilon(C)]$ for each vertex $v \in V'$ (note that here $\epsilon_v$ is the actual adjusted value for a vertex $v$ in $H$ which is corresponding to the component $C_v$ and this correct adjustment of $\epsilon_v$ to the dual variable $y_{C_v}$ is ensured by the Lemma~\ref{out-lemma1}). On the other hand maximum increment in the right hand side of the inequality is $(2 - \frac{1}{n - 1})\mathit{\sum_{v \in N_A} \epsilon_v}$. Therefore we can write,
  	\begin{align*} \mathit{\sum_{v \in N_A - N_D} \epsilon_vd_v + \sum_{v \in N_D}\epsilon_v} & \leq (2 - \frac{1}{n - 1})\mathit{\sum_{v \in N_A} \epsilon_v}
  	\end{align*}  
Writing the above inequality in details we get  	
  	\begin{align*}
  	\mathit{\sum_{v \in N_A - N_D} \epsilon_vd_v} & \leq (2 - \frac{1}{n - 1})\mathit{\sum_{v \in N_A - N_D}\epsilon_v} + (2 - \frac{1}{n - 1})\mathit{\sum_{v \in N_D}\epsilon_v - \sum_{v \in N_D}\epsilon_v}
	\end{align*} 
Since degree of each vertex in $\mathit{N_D}$ is zero, therefore the coefficient $(2 - \frac{1}{n - 1})$ of the term $(2 - \frac{1}{n - 1})\mathit{\sum_{v \in N_D}\epsilon_v}$ must be equal to 1. This implies the following inequality
	\begin{align*}
  	\mathit{\sum_{v \in N_A - N_D} \epsilon_vd_v} & \leq (2 - \frac{1}{n - 1})\mathit{\sum_{v \in N_A - N_D}\epsilon_v} 
	\end{align*} 
Rewriting the left hand side of the above inequality in terms of set $\mathit{N_A}$, $\mathit{N_I}$, and $\mathit{N_D}$ we get
	\begin{equation} \label{DPCST-eq3}
  	\mathit{\sum_{v \in N_A - N_D} \epsilon_vd_v  \leq \sum_{v \in (N_A - N_D) \cup N_I} \epsilon_vd_v - \sum_{v \in N_I}\epsilon_vd_v}
  	\end{equation} 	
	
	Before continuing with the proof we show that in graph $H$ there can be at most one leaf vertex of type $I$ which is corresponding to the component containing $r$. Suppose by contradiction $v \in V'$ is a leaf vertex of type $I$ in graph $H$  which is not the root vertex containing $r$ and an edge $e$ incidents on $v$ such that $\mathit{SE(e) = branch}$. Let $\mathit{C_v}$ be the {\em inactive} component corresponding to the vertex $v$. Since $C_v$ is a leaf of $H$ therefore the edge $\mathit{e \in F'}$. Note that after the termination of the {\em pruning} phase, $\mathit{F'}$ is the set of {\em branch} edges selected for the PCST. Since the state of $\mathit{C_v}$ is $\mathit{inactive}$ and it does not contain $r$, therefore it is deactivated at some point of execution of the algorithm and $\mathit{labelled\_flag = TRUE}$ for each $u \in C_v$. Furthermore, since $\mathit{C_v}$ is a leaf component, therefore no node $\mathit{u \in C_v}$ can be an intermediate node on the path of $\mathit{branch}$ edges between the vertex $r$ and a vertex of the status $\mathit{labelled\_flag = FALSE}$. Since $\mathit{labelled\_flag = TRUE}$ for each node $u \in C_v$ and $\mathit{C_v}$ is an $\mathit{inactive}$ leaf component therefore by the {\em pruning} phase of the algorithm each node $\mathit{u \in C_v}$ is pruned and $\mathit{SE(e') = basic}$ for each edge $e'$ such that $e' \in \delta(u)$. In this case one of the $e'$ must be $e$ such that $\mathit{SE(e) = basic}$, a contradiction to the fact that $\mathit{SE(e)= branch}$. Therefore except the root vertex, all other vertex of type $I$ are non leaf vertex in graph $H$. This fact implies that sum of degrees of all vertices of type $I$ in graph $H$ is at least $2\mathit{|N_I|} - 1$.

Since $\mathit{\epsilon_v} \in (0, \mathit{\epsilon(C)]}$, therefore replacing each $\mathit{\epsilon_v}$ by $\mathit{\epsilon(C)}$ we get the inequality (\ref{DPCST-eq3}) as follows
  	\begin{align*}
  	\mathit{\sum_{v \in N_A - N_D} \epsilon(C) d_v} & \leq 2 \mathit{\epsilon(C) (|(N_A - N_D) \cup N_I|} - 1) -  \mathit{\epsilon(C)} (2\mathit{|N_I|} - 1)
  	\end{align*}
  	(In the above inequality we use the fact that sum of degrees of all vertices is $2m$ where $m$ is the total number of edges in the graph.)
  	 
Since $(N_A - N_D)$ and $N_I$ are disjoint, therefore $|(N_A - N_D) \cap N_I| = |\phi| = 0$. Using this fact we get
  	\begin{align*}
  	\mathit{\sum_{v \in N_A - N_D} d_v} \leq &  2(\mathit{|(N_A - N_D})|) + 2\mathit{|N_I|} - 2 - 2\mathit{|N_I|} + 1\\
  	 = &  2(\mathit{|N_A - N_D|}) - 1\\
   	 = &  (2 - \frac{1}{\mathit{|N_A - N_D|}})|N_A - N_D|\\
    	 \leq &  (2 - \frac{1}{n - 1})\mathit{|N_A - N_D|}
  	\end{align*}
  	The last inequality holds since the number of type $A$ components is at most $n - 1$ for $n$ node graph, i.e. $\mathit{|N_A - N_D|} \leq (n - 1)$. Similarly it can be shown that the inequality (\ref{DPCST-eq1}) also holds for the case $\mathit{\epsilon(C)} \leq 0$. Therefore the inequality (\ref{DPCST-eq1}) always holds for every computed value $\mathit{\epsilon(C)}$ by a component $\mathit{C}$ in the graph.  
 \end{proof}

\subsection{Deadlock issue}
We show here that deadlock does not exist in D-PCST algorithm. Except the $\langle connect \rangle$, upon receiving any other messages a node can instantly reply or proceed with further actions of the algorithm. For example whenever a node $u$ receives $\langle test \rangle$ then $u$ immediately replies with $\langle status \rangle$ or $\langle reject \rangle$ by using its own local information and does not wait for any other event to be occurred on some other nodes. Now consider the case of {\em merging} of two neighboring components say $\mathit{C}$ and $\mathit{C'}$. This is only the case where a component needs to wait for another component to proceed further. Let the component $\mathit{C}$ sends $\langle connect \rangle$ to merge with the component $\mathit{C'}$. Upon receiving $\langle connect \rangle$, $\mathit{C'}$ responses to $\mathit{C}$ within a finite delay. If $\mathit{CS(C') = inactive}$ then $\mathit{C'}$ immediately sends $\langle accept \rangle$ without any further delay. If $\mathit{CS(C') = sleeping}$, then first it changes $\mathit{CS(C')}$ to $\mathit{active}$ and then finds its $\mathit{\epsilon(C')}$ and depending on  $\mathit{\epsilon(C')}$ it sends $\langle accept \rangle$ or $\langle refind\_epsilon \rangle$ to $\mathit{C}$. Since $\mathit{C'}$ does not depends on any event that delays the process of finding its $\mathit{\epsilon(C')}$ therefore $\mathit{C'}$ can response to $\mathit{C}$ within a finite delay which omits the possibility of any deadlock in between $\mathit{C'}$ and $\mathit{C}$. And in our proposed D-PCST algorithm, component grows only by sequential merging and no concurrent merging is allowed to happen. All of these observations ensure that deadlocks do not exist.

\section{Conclusion} \label{conclusion}
In this paper we propose D-PCST, the first asynchronous distributed deterministic algorithm for the PCST problem having an approximation factor of $(2 - \frac{1}{n - 1})$ of the optimal. Our algorithm is based on the sequential Goemans and Williamson algorithm (GW-algorithm)~\cite{GW_1995}. Compared to a trivial distributed implementation of the GW-algorithm using convergecast and broadcast (as, e.g.\ in~\cite{Rossetti_2015}, also in Subsubsection~\ref{subsec:message_complexity}) D-PCST behaves slightly worse in the worst case. However it does not suffer from single point of failure problems and requires only local message exchange in order to compute the PCST. Moreover, it has a better complexity in sparse or low diameter graphs. 
Since D-PCST is distributed in nature we believe that it can serve as a first step and a basis for further improvements of the 
the message and time complexity, as well as of the approximation ratio. In particular, we would like to 
investigate the applicability of our techniques in the direction of obtaining a distributed version of the PTAS of Bateni et al. \cite{Bateni:2011:PSP:2133036.2133115} for PCST in planar graphs; such a result would be of great theoretical and practical interest.



\begin{thebibliography}{100}

\bibitem{AA_AF_BG_2016}
Ahmad Abdi, Andreas~Emil Feldmann, and Bertrand Guenin.
\newblock {Lehman's Theorem and the Directed Steiner Tree Problem}.
\newblock {\em SIAM Journal of Discrete Mathematics}, 30(1):144--153, 2016.

\bibitem{AA_PK_RR_1995}
Ajit Agrawal, Philip Klein, and R.~Ravi.
\newblock When trees collide: An approximation algorithm for the generalized
  {S}teiner problem on networks.
\newblock {\em SIAM journal of Computing}, 24(3):440--456, 1995.

\bibitem{Zelikovsky1993}
Zelikovsky Alexander.
\newblock {An 11/6-approximation algorithm for the network Steiner problem}.
\newblock {\em Algorithmica}, 9(5):463--470, 1993.

\bibitem{EM_AC_XC_XH_BL_2010}
Eduardo \'{A}lvarez Miranda, A.~Candia, X.~Chen, X.~Hu, and B.~Li.
\newblock {Efficient Algorithms for the Prize Collecting Steiner Tree Problems
  with Interval Data}.
\newblock {\em International Conference on Algorithmic Applications in
  Management}, pages 13--24, 2010.

\bibitem{EM_IL_PT_2013}
Eduardo \'{A}lvarez Miranda, Ivana Ljubi\'{c}, and Paolo Toth.
\newblock {Exact approaches for solving robust prize-collecting Steiner tree
  problems}.
\newblock {\em European Journal of Operational Research}, 229(3):599--612,
  2013.

\bibitem{AA_MB_MH_2011}
Aaron Archer, Mohammad~Hossein Bateni, and Mohammad~Taghi Hajiaghayi.
\newblock {Improved Approximation Algorithms for Prize-Collecting Steiner Tree
  and TSP}.
\newblock {\em SIAM Journal on Computing}, 40(2):309--332, 2011.

\bibitem{Arora:1998:PTA:290179.290180}
Sanjeev Arora.
\newblock {Polynomial Time Approximation Schemes for Euclidean Traveling
  Salesman and Other Geometric Problems}.
\newblock {\em J. ACM}, 45(5):753--782, 1998.

\bibitem{Awerbuch:1987:ODA:28395.28421}
B.~Awerbuch.
\newblock Optimal {D}istributed {A}lgorithms for {M}inimum {W}eight {S}panning
  {T}ree, {C}ounting, {L}eader {E}lection, and {R}elated {P}roblems.
\newblock In {\em Proceedings of the Nineteenth Annual ACM Symposium on Theory
  of Computing}, pages 230--240, 1987.

\bibitem{BALAS_1989}
Egon Balas.
\newblock {The prize collecting traveling salesman problem}.
\newblock {\em Networks}, 19:621--636, 1989.

\bibitem{Bateni:2011:PSP:2133036.2133115}
M.~Bateni, C.~Chekuri, A.~Ene, M.~T. Hajiaghayi, N.~Korula, and D.~Marx.
\newblock Prize-collecting steiner problems on planar graphs.
\newblock In {\em Proceedings of the Twenty-second Annual ACM-SIAM Symposium on
  Discrete Algorithms}, SODA '11, pages 1028--1049, 2011.

\bibitem{FB_AV_1996}
Fred Bauer and Anujan Varma.
\newblock {Distributed Algorithms for Multicast Path Setup in Data Networks}.
\newblock {\em IEEE/ACM Transaction Networks}, 4(2):181--191, 1996.

\bibitem{BERMAN1994381}
P.~Berman and V.~Ramaiyer.
\newblock {Improved Approximations for the Steiner Tree Problem}.
\newblock {\em Journal of Algorithms}, 17(3):381 -- 408, 1994.

\bibitem{DB_MG_DS_DW_1993}
Daniel Bienstock, Michel~X. Goemans, David Simchi-Levi, and David Williamson.
\newblock A note on the prize collecting traveling salesman problem.
\newblock {\em Mathematical Programming}, 59:413--420, 1993.

\bibitem{Byrka:2010:ILA:1806689.1806769}
Jaroslaw Byrka, Fabrizio Grandoni, Thomas Rothvoss, and Laura Sanita.
\newblock {An Improved LP-based Approximation for Steiner Tree}.
\newblock In {\em Proceedings of the Forty-second ACM Symposium on Theory of
  Computing}, STOC '10, pages 583--592, 2010.

\bibitem{CAUNTO_2001}
S.A. Canuto, M.G.C. Resende, and C.C. Ribeiro.
\newblock {Local search with perturbation for prize-collecting Sreiner tree
  problem in graphs}.
\newblock {\em Networks}, 38(1):50--58, 2001.

\bibitem{PC_JF_2005}
Parinya Chalermsook and Jittat Fakckeroenphol.
\newblock {Simple distributed algorithm for approximating minimum Steiner
  tree}.
\newblock {\em The 11st International Computing and Combinatorics Conference
  (COCOON)}, pages 380--389, 2005.

\bibitem{Charikar:1998:AAD:314613.314700}
Moses Charikar, Chandra Chekuri, To-yat Cheung, Zuo Dai, Ashish Goel, Sudipto
  Guha, and Ming Li.
\newblock {Approximation Algorithms for Directed Steiner Problems}.
\newblock In {\em Proceedings of the Ninth Annual ACM-SIAM Symposium on
  Discrete Algorithms}, SODA '98, pages 192--200, 1998.

\bibitem{chlebik:2008:STP:1414105.1414423}
Miroslav Chleb\'{\i}k and Janka Chleb\'{\i}kov\'{a}.
\newblock {The Steiner Tree Problem on Graphs: Inapproximability Results}.
\newblock {\em Theoretical Computer Science}, 406(3):207--214, 2008.

\bibitem{Dittrich_2008}
Marcus~T. Dittrich, Gunnar~W. Klau, Andreas Rosenwald, Thomas Dandekar, and
  Tobias Muller.
\newblock {Identifying functional modules in protein-protein interaction
  networks: an integrated exact approach}.
\newblock {\em Intelligent Systems for Molecular Biology (ISMB)},
  24(13):119--141, 2008.

\bibitem{Du:2008:STP:1628718}
Dingzhu Du and Xiaodong Hu.
\newblock {\em {Steiner Tree Problems In Computer Communication Networks}}.
\newblock World Scientific Publishing Co., Inc., 2008.

\bibitem{Faloutsos:1995:ODA:224964.225474}
Michalis Faloutsos and Mart Molle.
\newblock {O}ptimal {D}istributed {A}lgorithm for {M}inimum {S}panning {T}rees
  {R}evisited.
\newblock In {\em Proceedings of the Fourteenth Annual ACM Symposium on
  Principles of Distributed Computing}, PODC '95, pages 231--237, 1995.

\bibitem{PF_CF_CF_JP_2007}
Paulo Feofiloff, Cristina~G. Fernandes, Carlos~E. Ferreira, and Jose~Coelho
  de~Pina.
\newblock {Primal-Dual Approximation Algorithms for the Prize-Collecting
  Steiner Tree Problem}.
\newblock {\em Information Processing Letter}, 103(5):195--202, 2007.

\bibitem{GHS_1983}
Robert~G. Gallager, Pierre~A. Humblet, and P.~M. Spira.
\newblock {A Distributed Algorithm for Minimum-Weight Spanning Trees}.
\newblock {\em ACM Transactions on Programming Languages and Systems},
  1:66--77, 1983.

\bibitem{GENHUEY199373}
Chen Gen-Huey, Michael~E. Houle, and Kuo Ming-Ter.
\newblock {The Steiner problem in distributed computing systems}.
\newblock {\em Information Sciences}, 74(1):73 -- 96, 1993.

\bibitem{GW_1995}
Michel~X. Goemans and David~E. Williamson.
\newblock A general approximation technique for constrained forest problems.
\newblock {\em SIAM Journal on Applied Mathematics}, 24(2):296--317, 1995.

\bibitem{Hanan_1966}
M.~Hanan.
\newblock {On Steiner's Problem with Rectilinear Distance}.
\newblock {\em SIAM Journal on Applied Mathematics}, 14(2):255--265, 1966.

\bibitem{Haouari_2010}
Mohamed Haouari, Safa~Bhar Layeb, and Hanif~D. Sherali.
\newblock {Algorithmic expedients for the Prize Collecting Steiner Tree
  Problem}.
\newblock {\em Discrete Optimization}, 7:32--47, 2010.

\bibitem{Hauptman-Karpinaski}
M.~Hauptmann and M.~Karpinski.
\newblock {A Compendium on {S}teiner Tree Problems}.
\newblock {\em
  http://theory.cs.uni-bonn.de/info5/steinerkompendium/netcompendium.html},
  Retrived April 2015.

\bibitem{Hwang_1976}
F.~K. Hwang.
\newblock {On Steiner Minimal Trees with Rectilinear Distance}.
\newblock {\em {SIAM Journal on Applied Mathematics}}, 30(1):104--114, 1976.

\bibitem{DJ_MM_SP_2000}
David~S. Johnson, Maria Minkoff, and Steven Phillips.
\newblock {The Prize Collecting Steiner Tree Problem: Theory and Practice}.
\newblock {\em Proceedings of the Eleventh Annual ACM-SIAM Symposium on
  Discrete Algorithms (SODA '00)}, pages 760--769, 2000.

\bibitem{DBLP:conf/coco/Karp72}
Richard~M. Karp.
\newblock Reducibility {A}mong {C}ombinatorial {P}roblems.
\newblock In {\em Proceedings of a Symposium on the Complexity of Computer
  Computations}, pages 85--103, 1972.

\bibitem{MK_FK_DM_GP_KT_2008}
Maleq Khan, Fabian Kuhn, Dahlia Malkhi, Gopal Pandurangan, and Kunal Talwar.
\newblock {Efficient Distributed Approximation Algorithms via Probabilistic
  Tree Embeddings}.
\newblock {\em Proceedings of the Twenty-seventh ACM Symposium on Principles of
  Distributed Computing(PODC '08)}, pages 263--272, 2008.

\bibitem{Klau2004}
Gunnar~W. Klau, Ivana Ljubi{\'{c}}, Andreas Moser, Petra Mutzel, Philipp
  Neuner, Ulrich Pferschy, G{\"u}nther Raidl, and Ren{\'e} Weiskircher.
\newblock Combining a memetic algorithm with integer programming to solve the
  prize-collecting steiner tree problem.
\newblock {\em Genetic and Evolutionary Computation (GECCO '04)}, pages
  1304--1315, 2004.

\bibitem{VPK_JCP_GCP_1993}
Vachaspathi~P. Kompella, Joseph~C. Pasquale, and George~C. Polyzos.
\newblock {Two Distributed Algorithms for Multicasting Multimedia Information}.
\newblock {\em IEEE/ACM Transaction on Networking}, 1:286--292, 1993.

\bibitem{citeulike:4031585}
Joseph~B. Kruskal.
\newblock {On the Shortest Spanning Subtree of a Graph and the Traveling
  Salesman Problem}.
\newblock {\em Proceedings of the American Mathematical Society}, 7(1):48--50,
  1956.

\bibitem{CL_BPS_2014}
Christoph Lenzen and Boaz Patt-Shamir.
\newblock {Improved Distributed Steiner Forest Construction}.
\newblock {\em In Proceedings of the 33th Annual ACM Symposium on Principles of
  Distributed Computing(PODC '14)}, pages 262--271, 2014.

\bibitem{TM_RW_1984}
T.~L. Magnanti and R.~T. Wong.
\newblock {Network design and transportation planning: models and algorithms}.
\newblock {\em Transportation Science}, 18:1--55, 1984.

\bibitem{Karpinski1997}
Karpinski Marek and Zelikovsky Alexander.
\newblock {New Approximation Algorithms for the Steiner Tree Problems}.
\newblock {\em Journal of Combinatorial Optimization}, 1(1):47--65, 1997.

\bibitem{JK4069504}
J.~K. Ousterhout, G.~T. Hamachi, R.~N. Mayo, W.~S. Scott, and G.~S. Taylor.
\newblock {The Magic VLSI Layout System}.
\newblock {\em IEEE Design Test of Computers}, 2(1):19--30, 1985.

\bibitem{BLTJ:BLTJ1515}
R.~C. Prim.
\newblock Shortest {C}onnection {N}etworks {A}nd {S}ome {G}eneralizations.
\newblock {\em Bell System Technical Journal}, 36(6):1389--1401, 1957.

\bibitem{Prodon2010}
Alain Prodon, Scott DeNegre, and Thomas~M. Liebling.
\newblock {Locating leak detecting sensors in a water distribution network by
  solving prize-collecting Steiner arborescence problems}.
\newblock {\em Mathematical Programming}, 124(1):119--141, 2010.

\bibitem{PROMEL200089}
Hans~Jurgen Promel and Angelika Steger.
\newblock {A New Approximation Algorithm for the Steiner Tree Problem with
  Performance Ratio 5/3}.
\newblock {\em Journal of Algorithms}, 36(1):89 -- 101, 2000.

\bibitem{Robins:2000:IST:338219.338638}
Gabriel Robins and Alexander Zelikovsky.
\newblock {Improved Steiner Tree Approximation in Graphs}.
\newblock In {\em Proceedings of the Eleventh Annual ACM-SIAM Symposium on
  Discrete Algorithms}, SODA '00, pages 770--779, 2000.

\bibitem{Robins:2005:TBG:1068396.1071708}
Gabriel Robins and Alexander Zelikovsky.
\newblock {Tighter Bounds for Graph Steiner Tree Approximation}.
\newblock {\em SIAM J. Discret. Math.}, 19(1):122--134, 2005.

\bibitem{Rossetti_2015}
Niccolo~G. Rossetti and P{\'a}ll Melsted.
\newblock {A First Attempt on the Distributed Prize-Collecting Steiner Tree
  Problem}.
\newblock {\em https://skemman.is/bitstream/1946/23105/1/niccolo.pdf}, 2015.

\bibitem{Wald_Colbourn_1983}
Joseph~A. Wald and Charles~J. Colbourn.
\newblock Steiner trees, partial 2–trees, and minimum {IFI} networks.
\newblock {\em Networks}, 13:159--167, 1983.

\bibitem{Watel:2013:SPL:2694605.2694640}
Dimitri Watel, Marc-Antoine Weisser, Cedric Bentz, and Dominique Barth.
\newblock {Steiner Problems with Limited Number of Branching Nodes}.
\newblock In {\em 20th International Colloquium on Structural Information and
  Communication Complexity (SIROCCO)}, pages 310--321, 2013.

\bibitem{Williamson1995}
David~P. Williamson, Michel~X. Goemans, Milena Mihail, and Vijay~V. Vazirani.
\newblock {A primal-dual approximation algorithm for generalized {S}teiner
  network problems}.
\newblock {\em Combinatorica}, 15(3):435--454, 1995.

\bibitem{Zosin:2002:DST:545381.545388}
Leonid Zosin and Samir Khuller.
\newblock {On Directed Steiner Trees}.
\newblock In {\em Proceedings of the Thirteenth Annual ACM-SIAM Symposium on
  Discrete Algorithms}, SODA '02, pages 59--63, 2002.

\end{thebibliography}

\appendix
\vspace{-2.2em}
\section{Description of the centralized Goemans-Williamson PCST algorithm \cite{GW_1995}}
\vspace{-.4em}
Since our algorithm is inspired by the centralized PCST algorithm proposed by Goemans and Williamson (GW-algorithm) \cite{GW_1995}, here we briefly describe the algorithm. GW-algorithm consists of two phases namely {\em growth} phase and {\em pruning} phase. The growth phase maintains a forest $F$ which contains a set of candidate edges being selected for the construction of the PCST. Initially $F$ is empty, each node is unmarked, and each node is considered as a connected component containing a singleton node. The growth phase also maintains a set of components whose possible states can be either active or inactive. If a component $C$ is active then the current state of $C$ is set to $1$, i.e., $CS(C)=1$, otherwise $CS(C)=0$. 
The state of the component containing $r$ is always inactive. Initially, except the root component, all other components are in active state. Associated with each component $C$, there is a dual variable $y_C$. Each $y_C$ is initialized to $0$. The GW-Algorithm also maintains a deficit value $d_v$ for each vertex $v \in V$ and a weight $W(C)$ for each component $C$. In each iteration the algorithm finds an edge $e = (u, v)$ with $u\in C_p$, $v \in C_q$, $C_p \neq C_q$, that minimizes $\epsilon_1 = \frac{w_e - d_v - d_u}{CS(C_p) + CS(C_q)}$ and a $C$ such that $CS(C) = 1$ which minimizes $\epsilon_2 = \sum_{v \in C}p_v - W(C)$. And then it finds the global minimum $\epsilon = min (\epsilon_1, \epsilon_2)$. Depending on the  value of $\epsilon$, the algorithm may decide to do any one of the two operations: (i) if $\epsilon = \epsilon_1$ then it merges two distinct components $C_p$ and $C_q$ using the edge $e$ (that gave the min $\epsilon$) and adds $e$ to $F$.   (ii) if $\epsilon = \epsilon_2$ then the corresponding component $C$ is deactivated. Note that for every decided value of $\epsilon$, for each component $C$ where $CS(C)=1$, the value of $W(C)$ (as well as the implicit $y_C$) and the value of each $d_v : v \in C$ is increased by the value of $\epsilon$. In case of merging, if the resulting component contains the root $r$ then it becomes inactive; otherwise it is active. In the other case i.e. deactivation of component $C$, the algorithm marks each $v \in C$ with the name of the component $C$. 
Since in each iteration of the algorithm, sum of the total number of components or the number of active components decreases therefore after at most $2n -1$  iterations all components become inactive. In pruning phase the algorithm removes as many edges as possible from $F$ without violating the two properties: (i) all unmarked vertices must be connected to the root, as these vertices never appeared in any deactivated components (ii) if a vertex with mark $C$ is connected to the root then every vertex marked with $C' \supseteq C$ should be connected to the root. The GW-algorithm achieves an approximation ratio of $(2 - \frac{1}{n-1})$  and running time of $O(n^2 \log n)$ for a graph of $n$ vertices.

\vspace{-1.2em}
\section{Pseudocode of the D-PCST algorithm}
\vspace{-2em}
\begin{algorithm}[ht!]
\caption{ The D-PCST algorithm : pseudocode for node $v$} \label{alg:D-PCST-algo}
\begin{algorithmic}[1]
\State {Upon receiving no message} 
\State \textbf{execute procedure} $initialization()$
\If {$v = r$} \Comment {Spontaneous awaken of the root node}
	\State {$CS \gets inactive; root\_flag \leftarrow TRUE; prize\_flag \gets FALSE;$}
	\State \textbf{execute procedure} $proc\_initiate()$
\Else
	\State {$CS \gets sleeping; root\_flag \leftarrow FALSE; prize\_flag \gets TRUE;$}
\EndIf
\vspace{1.5em}
\Procedure{$initialization$}{$ $} 
	\ForAll {$e \in \delta(v)$} 
		\State {$SE(e) \gets basic; EPM(e) \gets FALSE;$}  \Comment {$EPM$ : Edge for Prune Message} 
		
	\EndFor 
	
	\State {$d_h \gets 0; d_v \gets 0; W \leftarrow 0; labelled\_flag \gets FALSE; prune\_msg\_count \gets 0; proceed\_in\_edge \gets \phi; proceed\_flag \gets FALSE; leader\_flag \gets FALSE; received\_ts \leftarrow \infty;$}
\EndProcedure
\vspace{1.5em}
\Procedure{$proc\_initiate$}{$ $}
\State $SN \leftarrow find; find\_count \leftarrow 0; best\_epsilon \leftarrow \infty; best\_edge \leftarrow \phi; LC \leftarrow v; TP \gets 0; PF \gets FALSE;  back\_edge \gets \phi; TS \leftarrow \infty;$ 
\ForAll {$e \in \delta(v)$} 
	\If {$SE(e) = branch$}
		\State {{\bf send} $\langle initiate(LC,SN) \rangle$ on $e$}  
		\State $find\_count \leftarrow find\_count + 1;$\Comment{Count the number of $\langle initiate \rangle$ that are sent}
	\EndIf
\EndFor
\If {$SN = find$}
	\State \textbf{execute procedure} $proc\_test()$
\EndIf
\EndProcedure
\algstore{myalg}
\end{algorithmic}
\end{algorithm}

\begin{algorithm}                     
\begin{algorithmic} [1]                   
\algrestore{myalg}
\State Upon receiving $\langle initiate(L,S) \rangle$ on edge $e$
\State $ SN \leftarrow S; find\_count \leftarrow 0;  best\_epsilon \leftarrow \infty; best\_edge \leftarrow \phi; LC \leftarrow L; TP \gets 0; PF \gets FALSE; back\_edge \gets \phi; in\_branch \leftarrow e;$ 
\ForAll {$e' \in \delta(v) : e' \neq e $} 
	\If {$SE(e') = branch$}
		\State { \bf send}  $\langle initiate(L,S) \rangle$ on $e'$
		\State $find\_count \leftarrow find\_count + 1;$ \Comment{Count the number of  $\langle initiate \rangle$ that are sent}
	\EndIf
\EndFor
\If {$S = find$}
	\State \textbf{execute procedure} $proc\_test()$
\EndIf
\vspace{1em}
\Procedure{$proc\_test$}{$ $}
\State $test\_count \leftarrow 0;$
\ForAll {$e \in \delta(v)$}
	\If {$SE(e) = basic$ {\bf or} $SE(e) = refind$}
		\State {\bf send} $\langle test(LC) \rangle$ on $e$
		\State $test\_count \gets test\_count + 1;$	\Comment {Count the number of  $\langle test \rangle$ that are sent}
	\EndIf
\EndFor 
\EndProcedure
\vspace{1em}
\State {Upon receiving $\langle test(L) \rangle$ on edge $e$}
\If {$LC = L$}
	\State {\bf send} $\langle reject \rangle$ on $e$
\Else
	\State {\bf send} $\langle status(CS, d_v) \rangle$ on $e$.
\EndIf
\vspace{1em}
\State {Upon receiving $\langle status(NS, d_u) \rangle$ on edge $e$} \Comment{$u$ is a node belongs to a neighboring component}
\State $test\_count \gets test\_count - 1;$	
\If { $CS = active$ {\bf and} $NS = sleeping$}
	\State {$\epsilon_1 \gets \frac{w_e - d_v - d_h}{2}$;}
\ElsIf {$CS = active$ {\bf and}  $NS = inactive$ }
	\State {$\epsilon_1 \gets w_e- d_v - d_u$;}
\ElsIf { $CS = inactive$  {\bf and} $NS = sleeping$}                            
	\State {$\epsilon_1 \gets w_e - d_v - d_h$;}
\ElsIf {$CS = inactive$ {\bf and} $NS = inactive$ }
	\If {$SE(e) = refind$}
		\State {$\epsilon_1 \gets w_e- d_v - d_u$;}
	\Else
		\State {$\epsilon_1 \gets \infty$;}
	\EndIf
\EndIf
\If {$\epsilon_1  < best\_epsilon$}
	\State {$best\_epsilon \gets \epsilon_1; best\_edge \gets e;$}
\EndIf
\algstore{myalg}
\end{algorithmic}
\end{algorithm}

\begin{algorithm}                     
\begin{algorithmic} [1]                   
\algrestore{myalg}	
\State \textbf{execute procedure} $proc\_report()$
\vspace{1em}
\State {Upon receiving $\langle reject \rangle$ on edge $e$}
\State $test\_count \gets test\_count - 1;$
	\State {$SE(e) \gets rejected;$}
\If {$proceed\_in\_edge = e$} 				\Comment{The edge $e$ becomes a {\em rejected} edge}
	\State {$proceed\_in\_edge \gets \phi; proceed\_flag \gets FALSE;$}
\EndIf
\State \textbf{execute procedure} $proc\_report()$
\vspace{1em}
\Procedure{$proc\_report$} {$ $}
\If {$find\_count = 0$ {\bf and} $test\_count = 0$} \Comment {Receives responses for each $\langle initiate \rangle$ and $\langle test \rangle$}
	\State {$SN \gets found$;}
	\If {$d_h < d_v$}
		\State {$d_h \gets d_v;$}
	\EndIf
	\If {$CS = active$}
		\State {$TP \gets TP + p_v$} \Comment {$TP$: (Total Prize) of the subtree rooted at $v$}
	\EndIf

	\If {$proceed\_flag = TRUE$}
		\State {$PF \gets TRUE;$}
		\If {$TS > received\_ts$}
			\State{$TS \gets received\_ts; back\_edge \gets \phi;$}		
		\EndIf
\EndIf
	\If {$in\_branch \neq \phi$}
		\State {{\bf send} $\langle report(best\_epsilon, d_h, TP, PF, TS) \rangle$ on $in\_branch$}
	\Else
		\State \textbf{execute procedure} $proc\_merge\_or\_deactivate\_or\_proceed()$
	\EndIf
\EndIf
\EndProcedure

\vspace{1em}
\State {Upon receiving $\langle report(\epsilon_1, d_k, T, P, temp\_ts) \rangle$ on edge $e$}
\State {$find\_count = find\_count - 1;$}
\If {$P = TRUE$}
	\State {$PF \gets TRUE;$}
	\If {$TS > temp\_ts$}
		\State {$TS \gets temp\_ts; back\_edge \gets e;$}
	\EndIf
\EndIf
\If {$CS = active$}
	\State {$TP \gets TP + T;$}
\EndIf
\If {$d_h < d_k$}
	\State {$d_h \gets d_k;$}
\EndIf
\If {$\epsilon_1  < best\_epsilon$}
\algstore{myalg}
\end{algorithmic}
\end{algorithm}

\begin{algorithm}                     
\begin{algorithmic} [1]                   
\algrestore{myalg}
	\State {$best\_epsilon \gets \epsilon_1; best\_edge \gets e;$}
\EndIf
\State \textbf{execute procedure} $proc\_report()$
\vspace{1em}

\Procedure{$proc\_merge\_or\_deactivate\_or\_proceed$}{$ $}
\State {$\epsilon_1 \gets best\_epsilon$}
	\If {$root\_flag = FALSE$ {\bf and} $CS = active$}
		\State $\epsilon_2 \gets TP - W;$
		\If {$\epsilon_1 < \epsilon_2$}
			\State {$best\_epsilon \gets \epsilon_1$}
		\Else
			\State {$best\_epsilon \gets \epsilon_2$}
		\EndIf

		\If {$best\_epsilon = \epsilon_2$}
			\State {$CS \gets inactive; d_v \gets d_v + \epsilon_2; W \gets W + \epsilon_2; d_h \gets d_h + \epsilon_2; labelled\_flag \gets TRUE;$}

			
			
			\State {$deactivate\_flag \gets TRUE;$}	\Comment{$deactivate\_flag$ is a temporary variable}
			\State {{\bf send} $\langle update\_info(\epsilon_2, root\_flag, deactivate\_flag, W, d_h) \rangle$ on all $e \in \delta(v)$ such that $SE(e) = branch$} 
				\State \textbf{execute procedure} $proc\_initiate()$ \Comment{Compute $\epsilon_1$ to send $\langle proceed \rangle$ }

		\Else	\Comment{Start the merge procedure at the leader node}
			\If {$SE(best\_edge) = branch$} 
				\State {{\bf send} $\langle merge(best\_epsilon, d_h) \rangle$ on $best\_edge$}	
			\Else
				\State {{\bf send} $\langle connect(v, W, d_v, d_h) \rangle$ on $best\_edge$;}
			\EndIf
		\EndIf
	\ElsIf {$CS = inactive$}
		\If {($\epsilon_1 = \infty$)}
			\If {$TS = \infty$}
				\If {$v = r$} \hspace{4em} \Comment{Starts of pruning phase at the root node $r$.}
					\ForAll {$e \in \delta(v)$} 
						\If {$SE(e) = branch$ {\bf or} $EPM(e) = TRUE$} 
							\State {{\bf send} $\langle prune \rangle$ on $e$}
						\EndIf
						\If {$SE(e) = branch$} 
							\State {$prune\_msg\_count \gets prune\_msg\_count + 1;$}
						\EndIf
					\EndFor 
				\EndIf
			\Else			\Comment{Start of sending $\langle back \rangle$}
				\If {$back\_edge \neq \phi$}
					\State {{\bf send} $\langle back \rangle$ on $back\_edge$}
				\ElsIf {$back\_edge = \phi$ {\bf and} $proceed\_flag = TRUE$}
					\State {{\bf send} $\langle back \rangle$ on $proceed\_in\_edge$}
					\State {$proceed\_in\_edge \gets \phi; proceed\_flag = FALSE;$}
				\EndIf
			\EndIf
	\ElsIf {($\epsilon_1 \neq \infty$)}
\algstore{myalg}
\end{algorithmic}
\end{algorithm}

\begin{algorithm}                     
\begin{algorithmic} [1]                   
\algrestore{myalg}	
			\State {{\bf send} $\langle proceed(d_h) \rangle$ on $best\_edge$}
			\If {$SE(best\_edge) = basic$}
				\State {$EPM(best\_edge) \gets TRUE;$}
			\EndIf
			\If {$SE(best\_edge) = refind$}
				\State {$SE(best\_edge) \gets basic;$}
			\EndIf

		\EndIf
	\EndIf
\EndProcedure

\vspace{1em}
\State {Upon receiving $\langle merge(\epsilon, d_k) \rangle$ on edge $e$}
\If {$SE(best\_edge) = branch$} \Comment Receiving node is an intermediate node
	\State {{\bf send} $\langle merge(\epsilon, d_k) \rangle$ on $best\_edge$}
\Else 						 \Comment Receiving node is a frontier node	
		\State {{\bf send} $\langle connect(v, W, d_v, d_k) \rangle$ on $best\_edge$;}
		
\EndIf
\vspace{1em}
		
\State {Upon receiving $\langle back \rangle$ on edge $e$}
\If {$back\_edge \neq \phi$}
	\State {{\bf send} $\langle back \rangle$ on $back\_edge$}
\ElsIf {$back\_edge = \phi$ {\bf and} $proceed\_flag = TRUE$}
	\State {{\bf send} $\langle back \rangle$ on $proceed\_in\_edge$}
	\State {$proceed\_in\_edge \gets \phi; proceed\_flag = FALSE;$}
\ElsIf {$back\_edge = \phi$ {\bf and} $in\_branch \neq \phi$}
	\State {{\bf send} $\langle back \rangle$ on $in\_branch$}
\Else
	\State \textbf{execute procedure} $proc\_initiate()$
\EndIf

\vspace{1em}
\State {Upon receiving $\langle proceed(d_k) \rangle$ on edge $e$}
\If {$SE(e) = branch$ {\bf and} $in\_branch = e$}
	\State {{\bf send} $\langle proceed(d_k)  \rangle$ on $best\_edge$}
	\If {$SE(best\_edge) = basic$}
		\State {$EPM(best\_edge) \gets TRUE;$}
	\EndIf
	\If {$SE(best\_edge) = refind$}
		\State {$SE(best\_edge) \gets basic;$}
	\EndIf
\ElsIf {$SE(e) = basic$} 
	\State {$proceed\_flag \gets TRUE; proceed\_in\_edge \gets e;$} 
	\If {$CS = sleeping$}
		\State \textbf{execute procedure} $wakeup(d_k)$
	\ElsIf{$CS = inactive$}
		\If {$in\_branch \neq \phi$}
			\State {{\bf send} $\langle proceed(d_k)  \rangle$ on $in\_branch$}
		\EndIf
\algstore{myalg}
\end{algorithmic}
\end{algorithm}

\begin{algorithm}                     
\begin{algorithmic} [1]                   
\algrestore{myalg}	
	\EndIf
\ElsIf {$SE(e) = branch$ {\bf and} $in\_branch \neq e$}
	\If {$in\_branch \neq \phi$}
		\State {{\bf send} $\langle proceed(d_k)  \rangle$ on $in\_branch$}
	\Else			\Comment {Receiving node is the leader node}
		\State \textbf{execute procedure} $proc\_initiate()$
	\EndIf
\EndIf

\vspace{1em}
\Procedure{$wakeup$}{$d_k$}
    \State  {$CS \leftarrow active; d_v \gets d_k; W \leftarrow d_k;$}
	\If {$d_k > d_h$}
		\State {$d_h \gets d_k;$}
	\EndIf
\State \textbf{execute procedure} $proc\_initiate()$
\EndProcedure

\vspace{1em}
\State {Upon receiving $\langle connect(NID, WN, d_u, d_k) \rangle$ on edge $e$} \Comment {$WN$: Weight of Neighboring component}
\If {$CS = sleeping$}
	\State {$CS \leftarrow active; d_h \gets d_k; d_v \gets d_k; W \leftarrow d_k;$}
	\State {$\epsilon_1 \gets \frac{w_e - d_v - d_u}{2}; \epsilon_2 \gets p_v - W;$}	
	\If {$\epsilon_1 < \epsilon_2$}
		\If {$v > NID$}
			\State {$leader\_flag \gets TRUE;$}
		\Else
			\State {$leader\_flag \gets FALSE;$}
		\EndIf
			
		\State {$d_h \gets d_h + \epsilon_1; d_v \gets d_v + \epsilon_1; W \gets W + WN + 2*\epsilon_1; SE(e) \gets branch;$}
		\State {{\bf send} $\langle accept(leader\_flag, root\_flag, W, d_h) \rangle$ on $e$}
		\If {$leader\_flag \gets TRUE$}
			\State \textbf{execute procedure} $proc\_initiate()$
		\EndIf
	\Else	\Comment {$\epsilon_1 \geq \epsilon_2$}
		\State {$CS \gets inactive; W \gets W + \epsilon_2; d_v \gets d_v + \epsilon_2; d_h \gets d_k + \epsilon_2; labelled\_flag \gets TRUE;$}
		\State {{\bf send} $\langle refind\_epsilon \rangle$ on $e$}
	\EndIf
\ElsIf {$CS = inactive$}
	\If {$root\_flag = TRUE$}
		\State {$leader\_flag \gets TRUE;$}
	\Else
		\State {$CS \gets active;$}
		\If {$v > NID$}
			\State {$leader\_flag \gets TRUE;$}
		\Else
			\State {$leader\_flag \gets FALSE;$}
		\EndIf
	\EndIf
\algstore{myalg}
\end{algorithmic}
\end{algorithm}                 
\begin{algorithm}                     
\begin{algorithmic}[1]  
\algrestore{myalg}	
	\State {$\epsilon_1 = w_e - d_v - d_u; W \gets W + WN + \epsilon_1;$}
	\State {$d_t \gets d_k + \epsilon_1;$} \Comment{$d_t$ is a temporary variable}
	\If {$d_h < d_t$}
		\State {$d_h \gets d_t;$}
	\EndIf
	\State {$deactivate\_flag = FALSE;$}		\Comment{$deactivate\_flag$ is a temporary variable}
	\State {{\bf send} $\langle update\_info(0, root\_flag, deactivate\_flag, W, d_h) \rangle$ on all $e' \in \delta(v): e' \neq e$ {\bf and} $SE(e') = branch$} 
	\State {$SE(e) \gets branch;$}
	\State {{\bf send} $\langle accept(leader\_flag, root\_flag, W, d_h) \rangle$ on $e$;}

		\If {($leader\_flag = TRUE$) }
			\State \textbf{execute procedure} $proc\_initiate()$
		\EndIf
\EndIf	
	
\vspace{1em}
\State {Upon receiving $\langle refind\_epsilon \rangle$ on edge $e$}
\If {$SE(e) = basic$}
	\State {$SE(e) \gets refind;$}
\Else
	\If {$in\_branch \neq \phi$}
		\State {{\bf send} $\langle refind\_epsilon \rangle$ on $in\_branch$}
	\Else
		\State \textbf{execute procedure} $proc\_initiate()$
	\EndIf
\EndIf	

\vspace{1em}

\State {Upon receiving $accept(LF, RF, TW, d_k)$} on edge $e$ \Comment {$TW$: Total Weight}
\State {$SE(e) \gets branch; root\_flag \gets RF; d_h \gets d_k; d_v \gets d_v + best\_epsilon; W \gets TW;$}
\If {$RF = TRUE$}
		\State {$CS \gets inactive;  prize\_flag = FALSE;$}
\Else
	\State {$CS \gets active;$}
\EndIf
\If {$proceed\_in\_edge = e$ {\bf and} $proceed\_flag = TRUE$}
	\State {$ proceed\_in\_edge \gets \phi; proceed\_flag \gets FALSE;$}
\EndIf
\State {$deactivate\_flag = FALSE;$}		\Comment{$deactivate\_flag$ is a temporary variable}
\State {{\bf send} $\langle update\_info(best\_epsilon, root\_flag, deactivate\_flag, TW, d_h) \rangle$ on all $e'\in \delta(v) : e' \neq e$ {\bf and} $SE(e') = branch$}
\If {$LF = FALSE$}
		\State \textbf{execute procedure} $proc\_initiate()$
\EndIf

\vspace{1em}
\State {Upon receiving $\langle update\_info(EV, RF, DF, TW, d_k) \rangle$} on edge $e$
\If {$RF = TRUE$ and $DF = FALSE$}
	\State {$CS \gets inactive; prize\_flag \gets FALSE;$}
\algstore{myalg}
\end{algorithmic}
\end{algorithm}                 
\begin{algorithm}                     
\begin{algorithmic}[1]  
\algrestore{myalg}	
\ElsIf {$RF = FALSE$ and $DF = TRUE$}
	\State {$CS \gets inactive; labelled\_flag \leftarrow TRUE;$}
\ElsIf{$RF = FALSE$ and $DF = FALSE$}
	\State {$CS \gets active;$}
\EndIf
\State {$root\_flag \gets RF; d_h \gets d_k; d_v \gets d_v + EV; W \gets TW;$}
\State {{\bf send} $\langle update\_info(EV, RF, DF, TW, d_k) \rangle$ on all $e' \in \delta(v) : e' \neq e$ {\bf and} $SE(e') = branch$}
\If {$v = r$}		\Comment{$r$ is the root node}
	\State \textbf{execute procedure} $proc\_initiate()$
\EndIf
\vspace{1em}
\State {Upon receiving $\langle prune \rangle$ on edge $e$}
\If {$root\_flag = TRUE$} 	\Comment {Pruning inside the root component}
	\If {($labelled\_flag = TRUE$) {\bf and} ($SE(e') = basic$ {\bf for each} $e' \in \delta(v) : e' \neq e$)} 
			\State {$prize\_flag \gets TRUE; root\_flag \gets FALSE;$}
			\State {{\bf send} $\langle backward\_prune \rangle$ on $e$}
			\State {$SE(e) \gets basic;$}
	\Else
		\ForAll {$e' \in \delta(v) : e'\neq e$}
			\If {$SE(e') = branch$ {\bf or} $EPM(e') = TRUE$}
				\State {{\bf send} $\langle prune \rangle$ on $e'$}
				\If{$SE(e') = branch$} 
					\State {$prune\_msg\_count \gets prune\_msg\_count + 1;$}
				\EndIf
			\EndIf
		\EndFor 
	\EndIf
\Else  		\Comment {Pruning inside non-root inactive component}	
	\ForAll {$e' \in \delta(v) : e' \neq e$} 
		\If {$SE(e') = branch$ {\bf or} $EPM(e') = TRUE$ }
			\State {{\bf send} $\langle prune \rangle$ on $e'$}
			\If {$SE(e') = branch$}
				\State {$SE(e') \gets basic;$}
			\EndIf
		\EndIf
	\EndFor 
\EndIf	
\vspace{1em}
\State {Upon receiving $\langle backward\_prune \rangle$ on edge $e$}
\State {$prune\_msg\_count \gets prune\_msg\_count - 1; SE(e) \gets basic;$}
\If{$labelled\_flag = TRUE$ {\bf and} $prune\_msg\_count = 0$}
	\If {$in\_branch \neq \phi$} 
		\State {$prize\_flag \gets TRUE; root\_flag \gets FALSE;$}
		\State {{\bf send} $\langle backward\_prune \rangle$ on $in\_branch$}
		\State {$SE(in\_branch) \gets basic;$}
	\EndIf
\EndIf
\end{algorithmic}
\end{algorithm}

\end{document}